\documentclass[final]{siamltex}

\newcommand{\pres}{{ \pt }}
\newcommand{\Ne}{ {N_e}}
\newcommand{\nb}{{{\bf n }}}

\newcommand{\test}{{{q}}}
\newcommand{\testv}{{\vb}}
\newcommand{\testmu}{{\qh}}

\newcommand{\phihalf}{{\phi^{\frac{1}{2}} }}
\newcommand{\phimhalf}{{\phi^{-\frac{1}{2}} }}
\newcommand{\pOmegah}{{\pOmega_h}}
\newcommand{\Omegah}{{\Omega_h}}
\def\d{\delta} 
 
\def\e{{\epsilon}}

\def\norm#1{\|#1\|}

\usepackage{amssymb}
\usepackage{amsmath}
\usepackage{float}
\usepackage[pdfborder={0 0 0}, colorlinks=true, linkcolor=blue, citecolor=red]{hyperref}
\hypersetup{pdfborder=0 0 0}
\usepackage{algorithm}
\usepackage{algorithmic}
\usepackage{mathrsfs}
\usepackage{subfigure}
\usepackage{leftidx}

\usepackage[latin1]{inputenc}
\usepackage{showkeys,cite}
\usepackage{ulem}

\usepackage{multirow}

\usepackage{color}
\usepackage{soul,xargs}
\usepackage[pdftex,dvipsnames]{xcolor}

\usepackage[colorinlistoftodos,prependcaption,textsize=tiny]{todonotes}

\newcommandx{\question}[2][1=]{\todo[linecolor=red,backgroundcolor=red!25,bordercolor=red,#1]{#2}}
\newcommandx{\change}[2][1=]{\todo[linecolor=blue,backgroundcolor=blue!25,bordercolor=blue,#1]{#2}}
\newcommandx{\add}[2][1=]{\todo[linecolor=OliveGreen,backgroundcolor=OliveGreen!25,bordercolor=OliveGreen,#1]{#2}}
\newcommandx{\improve}[2][1=]{\todo[linecolor=Plum,backgroundcolor=Plum!25,bordercolor=Plum,#1]{#2}}
\newcommandx{\thiswillnotshow}[2][1=]{\todo[disable,#1]{#2}}
\newcommandx{\remove}[2][1=]{\todo[linecolor=yelllow,backgroundcolor=yellow!10,bordercolor=red,#1]{#2}}

\usepackage[top=1in, bottom=1in, left=1in, right=1in]{geometry}

\begin{document}

\title{A Hybridized Discontinuous Galerkin Method for Linear Degenerate Elliptic Equation Arising from Two-Phase Mixtures\thanks{The first and second authors are 
    partially  supported by the NSF Grant NSF-DMS1620352. We are grateful for the support.}}


\author{Shinhoo Kang\footnotemark[1], Tan Bui-Thanh \footnotemark[2], and Todd Arbogast\footnotemark[3]}

\renewcommand{\thefootnote}{\fnsymbol{footnote}}

\footnotetext[3]{Department of Mathematics, and Institute for Computational Engineering
  \& Sciences, The University of Texas at Austin, Austin, TX 78712}

\footnotetext[2]{Department of Aerospace Engineering and Engineering Mechanics, and Institute for Computational Engineering
  \& Sciences, The University of Texas at Austin, Austin, TX 78712,
USA.}
\footnotetext[1]{Department of Aerospace Engineering and Engineering Mechanics,
  The University of Texas at Austin, Austin, TX 78712,
  USA.}

\renewcommand{\thefootnote}{\arabic{footnote}}

\newcommand{\TODO}[1]{ \fbox{\parbox{3in}{\bf TODO: #1}}}

\newcommand{\grbf}[1] {\mbox{\boldmath${#1}$\unboldmath}}
\newcommand{\gbf}[1] {\mathbf{#1}}

\newcommand{\beq} {\begin{equation}}
\newcommand{\eeq} {\end{equation}}
\newcommand{\bdm} {\begin{displaymath}}
\newcommand{\edm} {\end{displaymath}}
\newcommand{\bit}{\begin{itemize}}
\newcommand{\eit}{\end{itemize}}
\newcommand{\bde}{\begin{description}}
\newcommand{\ede}{\end{description}}
\newcommand{\bce}{\begin{center}}
\newcommand{\ece}{\end{center}}
\newcommand{\ben} {\begin{enumerate}}
\newcommand{\een} {\end{enumerate}}
\newcommand{\bea} {\begin{eqnarray}}
\newcommand{\eea} {\end{eqnarray}}
\newcommand{\barr} {\begin{array}}
\newcommand{\earr} {\end{array}}
\newcommand{\bean} {\begin{eqnarray*}}
\newcommand{\eean} {\end{eqnarray*}}
\newcommand{\edoc} {\end{document}}
\newcommand{\On}{\hspace{0.2cm} \mbox{on} \hspace{0.2cm}}
\newcommand{\In}{\hspace{0.2cm} \mbox{in} \hspace{0.2cm}}
\newcommand{\Minimize}{\hspace{0.2cm} \mbox{minimize} \hspace{0.2cm}}
\newcommand{\Maximize}{\hspace{0.2cm} \mbox{maximize} \hspace{0.2cm}}
\newcommand{\SubjectTo}{\hspace{0.2cm} \mbox{subject} \hspace{0.15cm}
            \mbox{to} \hspace{0.2cm}}

\def\etal{{\it et al.~}}

\newcommand{\point}[1] {\ensuremath{\boldsymbol{#1}}}
\newcommand{\vvec}[1] {\ensuremath{\boldsymbol{#1}}}
\renewcommand{\vec}[1] {\ensuremath{\boldsymbol{#1}}}
\newcommand{\cvec}[1] {\ensuremath{\boldsymbol{{#1}}}}
\newcommand{\ten}[1] {\ensuremath{\boldsymbol{#1}}}
\newcommand{\tenfour}[1] {\ensuremath{\boldsymbol{\mathsf{#1}}}}
\newcommand{\comp}[1]{\begin{pmatrix}{ #1 }\end{pmatrix}}
\newcommand{\vv}{\ensuremath{\vec{v}}}
\newcommand{\vr}{\left(\ensuremath{\vec{r}}\right)}
\newcommand{\att}{\left(t\right)}
\newcommand{\tvv}{\ensuremath{\tilde{\vec{v}}}}
\newcommand{\ww}{\ensuremath{\vec{w}}}
\newcommand{\GG}{\ensuremath{\ten{G}}}
\newcommand{\FF}{\ensuremath{\boldsymbol{\mathfrak{F}}}}
\newcommand{\tEE}{\ensuremath{\tilde{\vec{E}}}}
\newcommand{\nn}{\ensuremath{\vec{n}}}
\renewcommand{\SS}{\ensuremath{\ten{S}}}
\newcommand{\tSS}{\ensuremath{\tilde{\ten{S}}}}
\newcommand{\cp}{\ensuremath{{c_p}}}
\newcommand{\cs}{\ensuremath{{c_s}}}
\newcommand{\tcs}{\ensuremath{{\tilde{c}_s}}}
\newcommand{\Vp}{\ensuremath{{V^e_p}}}
\newcommand{\Vs}{\ensuremath{{V^e_s}}}
\newcommand{\Ss}{\ensuremath{{S^e_s}}}
\newcommand{\nxnx}[1] {\ensuremath{\nn\times\left(\nn\times{#1}\right)}}

\newcommand{\dd}[2] {\ensuremath{\frac{\partial {#1}}{\partial {#2}}}}
\newcommand{\DD}[2] {\ensuremath{\frac{d {#1}}{d {#2}}}}
\newcommand{\Grad} {\ensuremath{\nabla}}
\newcommand{\Gradv}[1]{\ensuremath{\nabla_{#1}}}
\newcommand{\Div} {\ensuremath{\nabla\cdot}}
\newcommand{\Divv}[1] {\ensuremath{\nabla_{#1}\cdot}}
\newcommand{\Curl} {\ensuremath{\nabla\times}}
\newcommand{\Cten}{\ensuremath{\tenfour{C}}}

\newcommand{\eval}[2][\right]{\relax
  \ifx#1\right\relax \left.\fi#2#1\rvert}

\providecommand{\abs}[1]{\left\lvert#1\right\rvert}
\providecommand{\norm}[1]{\left\lVert#1\right\rVert}

\newcommand{\Nel} {\ensuremath{{N_\text{el}}}}
\newcommand{\Ned} {\ensuremath{{N_\text{ed}}}}
\newcommand{\De} {\ensuremath{{\mathsf{D}^e}}}
\newcommand{\Dep} {\ensuremath{{\mathsf{D}^{e'}}}}
\newcommand{\Dhat} {\ensuremath{\hat{\mathsf{D}}}}
\newcommand{\Dehat} {\ensuremath{{\hat{\mathsf{D}}^e}}}
\newcommand{\half} {\ensuremath{\frac{1}{2}}}
\newcommand{\IN} {\ensuremath{{\mathcal{I}_N}}}
\newcommand{\INe} {\ensuremath{{\mathcal{I}_{N_e}}}}

\newcommand{\mc}[1]{\mathcal{#1}}
\newcommand{\mcC}[1]{\mathcal{#1}}
\newcommand{\mb}[1]{\mathbf{#1}}
\newcommand{\mbb}[1]{\mathbb{#1}}
\newcommand{\mi}[1]{\mathit{#1}}
\newcommand{\mf}[1]{\mathfrak{#1}}
\newcommand{\Oi}[1]{\int_{\mc{D}} #1 \, d\vec{x}}
\newcommand{\TOi}[1]{\int_{0}^T \int_{\mc{D}} #1 \, d\vec{x}dt}
\newcommand{\Ti}[1]{\int_{0}^T #1 \, dt}
\newcommand{\TDei}[1]{\int_{0}^T\int_{\De} #1 \, d\vec{x}dt}
\newcommand{\TDeid}[1]{\int_{0}^T\int_{\De,N} #1 \, d\vec{x}dt}
\newcommand{\Dei}[1]{\int_{\De} #1 \, d\vec{x}}
\newcommand{\DeI}[1]{\int_{D} #1 \, d\vec{x}}
\newcommand{\DeiM}[1]{\int_{\Dhat} #1 \, d\vec{r}}
\newcommand{\DeMd}[1]{\int_{\De,N_e} #1 \, d\vec{x}}
\newcommand{\DeiMd}[1]{\int_{\Dhat,N_e} #1 \, d\vec{r}}
\newcommand{\DeMdN}[1]{\int_{D,N} #1 \, d\vec{x}}
\newcommand{\DeiMdN}[1]{\int_{\Dhat,N} #1 \, d\vec{r}}
\newcommand{\DeiMdNC}[2]{\int_{\Dhat,N_{#2}} #1 \, d\vec{r}}
\newcommand{\pDei}[1]{\int_{\partial \De} #1 \, d\vec{x}}
\newcommand{\pDeiM}[1]{\int_{\partial \Dhat} #1 \, d\vec{r}}
\newcommand{\pDeiMd}[1]{\int_{\partial \Dhat,N_e} #1 \, d\vec{r}}
\newcommand{\pDeiMdN}[1]{\int_{\partial \Dhat,N} #1 \, d\vec{r}}
\newcommand{\pDeiMdNC}[2]{\int_{\partial \Dhat,N_{#2}} #1 \, d\vec{r}}
\newcommand{\pMortar}[2]{\int_{\mc{M}_{#2}} #1 \, d\vec{r}}
\newcommand{\pMortard}[2]{\int_{\mc{M}^{#2},N_{#2}} #1 \, d\vec{r}}
\newcommand{\pMortardd}[3]{\int_{\mc{M}^{#2},N_{#3}} #1 \, d\vec{r}}
\newcommand{\pMortardh}[3]{\int_{\mc{M}_{#2},N_{#3}} #1 \, d\vec{r}}
\newcommand{\pMortardhn}[4]{\int_{{#2}_{#3},N_{#4}} #1 \, d\vec{r}}
\newcommand{\TpDei}[1]{\int_{0}^T\int_{\partial \De} #1 \, d\vec{x}}
\newcommand{\TpDeid}[1]{\int_{0}^T\int_{\partial \De,N} #1 \, d\vec{x}}
\newcommand{\Deid}[1]{\int_{\De,N_e} #1 \, d\vec{x}}
\newcommand{\DeidN}[1]{\int_{\De,N} #1 \, d\vec{x}}
\newcommand{\pDeidN}[1]{\int_{\partial \De,N} #1 \, d\vec{x}}
\newcommand{\pDeid}[1]{\int_{\partial \De,N_e} #1 \, d\vec{x}}
\newcommand{\nor}[1]{\left\| #1 \right\|}
\newcommand{\nort}[1]{{\left\vert\kern-0.25ex\left\vert\kern-0.25ex\left\vert #1 
    \right\vert\kern-0.25ex\right\vert\kern-0.25ex\right\vert}}
\newcommand{\snor}[1]{\left| #1 \right|}
\newcommand{\LRp}[1]{\left( #1 \right)}
\newcommand{\LRs}[1]{\left[ #1 \right]}
\newcommand{\LRa}[1]{\left< #1 \right>}
\newcommand{\LRc}[1]{\left\{ #1 \right\}}
\newcommand{\pp}[2]{\frac{\partial #1}{\partial #2}}
\newcommand{\ppcp}[1]{\frac{\partial #1}{\partial \vec{c_p}}}
\newcommand{\ppcpj}[1]{\frac{\partial #1}{\partial c_p}}
\newcommand{\ppcpmf}[1]{\frac{\partial #1}{\partial \vec{\mf{c_p}}}}
\newcommand{\ppm}[1]{\frac{\partial #1}{\partial \vec{m}}}
\newcommand{\ppho}[3]{\frac{\partial^{#3} #1}{{\partial #2}^{#3}}}
\newcommand{\ppsom}[3]{\frac{\partial^{2} #1}{{\partial #2}{\partial #3}}}
\newcommand{\ppmi}[1]{\frac{\partial #1}{\partial m_i}}
\newcommand{\ppmj}[1]{\frac{\partial #1}{\partial m_j}}
\newcommand{\ppcpm}[1]{\frac{\partial #1}{\partial \vec{c_p}^-}}
\newcommand{\ppcpp}[1]{\frac{\partial #1}{\partial \vec{c_p}^+}}
\newcommand{\ppcs}[1]{\frac{\partial #1}{\partial \vec{c_s}}}
\newcommand{\ppcsm}[1]{\frac{\partial #1}{\partial \vec{c_s}^-}}
\newcommand{\ppcsp}[1]{\frac{\partial #1}{\partial \vec{c_s}^+}}
\newcommand{\td}[2]{\frac{d #1}{d #2}}
\newcommand{\Rvert}[1]{\left. #1 \right|}
\newcommand{\bs}[1]{\boldsymbol{#1}}
\newcommand{\ipDed}[3]{\LRp{ #1, #2}_{\De,N}^{#3}}

\newcommand{\bigO}{\mathcal{O}}

\newcommand{\iOm}[1]{\int_{\Omega} #1 \, d\Omega}
\newcommand{\iGb}[1]{\int_{\Gamma_{b}} #1 \, ds}
\newcommand{\iGs}[1]{\int_{\Gamma_{s}} #1 \, ds}
\newcommand{\iGsy}[1]{\int_{\Gamma_{s}} #1 \, ds(\mb{y})}
\newcommand{\RE}{\mbox{{I}}\hspace{-0.5ex}\mbox{{R}}}
\newcommand{\iC}[1]{\int_{0}^{2\pi} #1 \, d\theta}
\newcommand{\iGr}[1]{\int_{\Gamma_{\infty}} #1 \, ds}

\newcommand{\jump}[1]{{\llbracket #1 \rrbracket}} 
\newcommand{\diff}[1] {\ensuremath{\LRs{#1}}}
\newcommand{\average}[1] {\ensuremath{\LRc{\!\!\{#1\}\!\!}}}
\newcommand{\averageM}[1] {\ensuremath{\LRc{\hspace{-2pt}\LRc{#1}\hspace{-2pt}}}}

\newcounter{tablecell}

\newcommand*{\numbercell}{%
  \stepcounter{tablecell}%
  \thetablecell. 
}

\newtheorem{propo}{Proposition}
\newtheorem{coro}{Corollary}
\newtheorem{theo}{Theorem}
\newtheorem{lem}{Lemma}

\newtheorem{defi}{definition}

\newtheorem{rema}{remark}

\newcommand{\figlab}[1]{\label{fig:#1}}
\newcommand{\eqnlab}[1]{\label{eq:#1}}
\newcommand{\theolab}[1]{\label{theo:#1}}
\newcommand{\corolab}[1]{\label{coro:#1}}
\newcommand{\propolab}[1]{\label{propo:#1}}
\newcommand{\lemlab}[1]{\label{lem:#1}}
\newcommand{\defilab}[1]{\label{defi:#1}}
\newcommand{\remalab}[1]{\label{rema:#1}}
\newcommand{\tablab}[1]{\label{tab:#1}}

\newcommand{\figref}[1]{\ref{fig:#1}}
\newcommand{\theoref}[1]{\ref{theo:#1}}
\newcommand{\defiref}[1]{\ref{defi:#1}}
\newcommand{\remaref}[1]{\ref{rema:#1}}
\newcommand{\cororef}[1]{\ref{coro:#1}}
\newcommand{\proporef}[1]{\ref{propo:#1}}
\newcommand{\lemref}[1]{\ref{lem:#1}}
\newcommand{\eqnref}[1]{\eqref{eq:#1}}
\newcommand{\alglab}[1]{\label{alg:#1}}
\newcommand{\algref}[1]{\ref{alg:#1}}
\newcommand{\seclab}[1]{\label{sect:#1}}
\newcommand{\secref}[1]{\ref{sect:#1}}
\newcommand{\tabref}[1]{\ref{tab:#1}}

\renewcommand{\u}{u}
\renewcommand{\v}{v}
\newcommand{\un}{\u_h}
\renewcommand{\e}{e}
\newcommand{\w}{w}
\newcommand{\wb}{\mb{\w}}
\newcommand{\J}{J}
\newcommand{\Jb}{\mb{J}}
\newcommand{\q}{q}
\renewcommand{\r}{r}
\newcommand{\qh}{\hat{\q}}
\newcommand{\qs}{\q^*}
\newcommand{\qht}{\tilde{\qh}}
\newcommand{\qbar}{\overline{\q}}
\newcommand{\qb}{{\mb{\q}}}
\newcommand{\sig}{{\sigma}}
\newcommand{\thetah}{{\hat{\theta}}}
\newcommand{\thetas}{{\theta^*}}
\newcommand{\thetahn}{{\thetah_h}}
\newcommand{\sign}[1]{\text{ sgn}\!\LRp{#1}}
\newcommand{\sigh}{\hat{\sig}}
\newcommand{\sigs}{\sig^*}
\newcommand{\sigb}{{\bs{\sig}}}
\newcommand{\sigbn}{\sigb_h}
\newcommand{\sigbs}{\sigb^*}
\newcommand{\taub}{{\bs{\tau}}}
\newcommand{\kappab}{{\bs{\kappa}}}
\newcommand{\gamb}{{\bs{\gamma}}}
\newcommand{\ut}{\tilde{\u}}
\newcommand{\sigbt}{\tilde{\sigb}}
\newcommand{\vt}{\tilde{\v}}
\newcommand{\taubt}{\tilde{\taub}}
\newcommand{\ub}{\mb{\u}}
\newcommand{\ubp}{\mb{\u}^+}
\newcommand{\ubm}{\mb{\u}^-}
\newcommand{\vb}{\mb{\v}}
\newcommand{\um}{\u^-}
\newcommand{\up}{\u^+}

\newcommand{\m}{{m}}

\newcommand{\ud}{{\dot{\u}}}
\newcommand{\vd}{{\dot{\v}}}
\newcommand{\sigbd}{{\dot{\sigb}}}
\newcommand{\taubd}{\dot{\taub}}

\renewcommand{\d}{{d}}
\newcommand{\R}{{\mathbb{R}}}
\newcommand{\kap}{\kappa}
\newcommand{\F}{\mc{F}}
\newcommand{\f}{f}
\newcommand{\g}{g}
\newcommand{\fb}{\mb{\f}}
\newcommand{\gb}{\mb{\g}}
\newcommand{\gD}{{g_D}}
\newcommand{\pOmega}{\partial \Omega}
\newcommand{\pOh}{{\pOmega_h}}
\newcommand{\K}{K}
\newcommand{\Kp}{K^+}
\newcommand{\Km}{K^-}
\newcommand{\pK}{{\partial \K}}
\newcommand{\Kj}{{\K_j}}
\newcommand{\Gh}{{\mc{E}_h}}
\newcommand{\Gho}{{\mc{E}_h^o}}
\newcommand{\Ghb}{{\mc{E}_h^{\partial}}}
\newcommand{\Poly}{{\mc{P}}}
\newcommand{\poly}[1]{\mc{P}_{#1}}
\newcommand{\polyb}[1]{\bs{\mc{P}}_{#1}}
\newcommand{\D}{{D}}
\newcommand{\p}{{p}}
\newcommand{\pt}{\tilde{\p}}
\newcommand{\ph}{\hat{\p}}
\newcommand{\Ts}{{\mb{T}}}
\newcommand{\Vb}{{\mb{V}}}
\newcommand{\V}{{V}}
\newcommand{\Lam}{{\Lambda}}
\newcommand{\Lamb}{\bs{\Lambda}}
\newcommand{\mub}{\bs{\mu}}
\newcommand{\lambdab}{\bs{\lambda}}
\newcommand{\Th}{{\Ts_h}}
\newcommand{\Vh}{{\V_h}}
\newcommand{\Vbh}{{\Vb_h}}
\newcommand{\Lamh}{{\Lam_h}}
\newcommand{\Lambh}{{\Lamb_h}}
\newcommand{\Lamho}{{\overline{\Lam}_h}}
\newcommand{\Lambht}{{\Lamb_h^t}}
\newcommand{\Ltwo}{{L^2\LRp{\Omega}}}
\newcommand{\Ltwoe}{{L^2\LRp{\e}}}
\newcommand{\Ltw}{{L^2\LRp{\K}}}
\newcommand{\VhK}{{\V_h\LRp{\K}}}
\newcommand{\ThK}{{\Ts_h\LRp{\K}}}
\newcommand{\VbhK}{{\Vb_h\LRp{\K}}}
\newcommand{\Lamhe}{{\Lam_h\LRp{\e}}}
\newcommand{\Lambhe}{{\Lamb_h\LRp{\e}}}
\renewcommand{\L}{\mc{L}}
\renewcommand{\D}{\mc{D}}
\newcommand{\T}{\mc{T}}
\newcommand{\Tt}{\tilde{\T}}
\newcommand{\A}{\mb{A}}
\newcommand{\B}{\mb{B}}
\renewcommand{\D}{\mb{D}}
\newcommand{\Am}{\A^-}
\newcommand{\Ap}{\A^+}
\newcommand{\Aa}{\snor{\A}}
\newcommand{\Rb}{\mb{R}}
\newcommand{\C}{\mc{C}}
\renewcommand{\S}{\mb{S}}
\newcommand{\Sm}{\S^{-}}
\newcommand{\Sp}{\S^{+}}
\newcommand{\Spm}{\S^{\pm}}
\newcommand{\Q}{\mb{Q}}

\newcommand{\Lte}{{L^2\LRp{\Gh}}}
\newcommand{\n}{\mb{n}}
\newcommand{\nm}{\n^-}
\newcommand{\np}{\n^+}
\newcommand{\uh}{{\hat{\u}}}
\newcommand{\us}{{\u^*}}
\newcommand{\uhn}{{\uh_h}}
\newcommand{\ubh}{{\hat{\ub}}}
\newcommand{\ubs}{{\ub^*}}
\newcommand{\ubht}{{\tilde{\ubh}}}
\newcommand{\uht}{{\tilde{\uh}}}
\newcommand{\uhtn}{{\uht_h}}
\newcommand{\uhd}{{\dot{\uh}}}
\newcommand{\sigbh}{{\hat{\sigb}}}
\newcommand{\Fh}{{\hat{\F}}}
\newcommand{\Fs}{\F^*}
\renewcommand{\c}{{c}}

\newcommand{\zb}{\mb{z}}
\renewcommand{\sb}{\mb{s}}
\newcommand{\rb}{\mb{r}}
\newcommand{\betab}{\bs{\beta}}
\newcommand{\veps}{{\varepsilon}}
\newcommand{\vepsb}{\bs{\veps}}
\newcommand{\Hb}{\mb{H}}
\newcommand{\Hbt}{\Hb^t}
\newcommand{\Eb}{\mb{E}}
\newcommand{\Ebt}{\Eb^t}
\newcommand{\hb}{\mb{h}}
\newcommand{\Hbh}{\hat{\mb{H}}}
\newcommand{\Ebh}{\hat{\mb{E}}}
\newcommand{\Hbs}{{\mb{H}^*}}
\newcommand{\Ebs}{{\mb{E}^*}}
\newcommand{\Ebht}{\Ebh^t}
\newcommand{\Hbht}{\Hbh^t}
\newcommand{\Lb}{\mb{L}}
\newcommand{\Gb}{\mb{G}}
\newcommand{\Lbh}{\hat{\Lb}}
\newcommand{\Lbs}{\Lb^*}
\newcommand{\Ib}{\mb{I}}
\newcommand{\Sigb}{\bs{\Sigma}}
\newcommand{\Z}{Z}
\newcommand{\Y}{Y}
\newcommand{\rhou}{\rho\u}
\newcommand{\rhov}{\rho\v}
\newcommand{\E}{E}
\newcommand{\Escript}{\mathsf{\E}}
\newcommand{\h}{h}
\newcommand{\vel}{\bs{\vartheta}}
\newcommand{\vels}{\vel^*}
\newcommand{\vele}{\vel^e}
\newcommand{\velh}{\hat{\vel}}
\newcommand{\phis}{\phi^*}
\newcommand{\phie}{\phi^e}
\newcommand{\phih}{\hat{\phi}}
\newcommand{\phio}{\phi^0}
\newcommand{\uo}{\u^0}
\newcommand{\vo}{\v^0}
\newcommand{\vs}{\v^*}
\newcommand{\vh}{\hat{\v}}
\newcommand{\Ic}{\mc{I}}
\newcommand{\ue}{{\u^e}}
\newcommand{\ube}{{\ub^e}}
\newcommand{\ve}{{\v^e}}
\renewcommand{\b}{{b}}

\renewcommand{\k}{{k}}
\renewcommand{\a}{{a}}
\renewcommand{\pOmegah}{{\pOmega_h}}
\newcommand{\nortc}[1]{{\left]\kern-0.25ex\left\vert\kern-0.25ex\left\vert #1 
  \right\vert\kern-0.25ex\right\vert\kern-0.25ex\right[}}

\renewcommand{\mu}{{\qh}}
\newcommand{\gamD}{{\Gamma_D}}
\newcommand{\pe}{{\p^e}}
\newcommand{\phe}{\widehat{\pe}}
\newcommand{\Piu}{\Pi_\ube}
\newcommand{\Pip}{\Pi_\pe}
\renewcommand{\Pr}{\mathbb{P}}
\newcommand{\eub}{\bs{\varepsilon}_{\ub}}
\newcommand{\ep}{\varepsilon_{\p}}
\newcommand{\eph}{\varepsilon_{\ph}}
\renewcommand{\P}{\bs{P}}
\renewcommand{\Pr}{{\mathbb{P}}}
\newcommand{\Prb}{\bs{\Pr}}

\newcommand{\mygreen}[1]{{\color[rgb]{0,.65,0} #1}}
\newcommand{\mywhite}[1]{{\color[rgb]{1.0,1.0,1.0} #1}}
\newcommand{\myred}[1]{{\color[rgb]{0.65,0.0,0.0} #1}}
\newcommand{\myblue}[1]{{\color[rgb]{0.0,0.0,1.0} #1}}
\newcommand{\mygray}[1]{{\color[rgb]{.15,.35,.60} #1}}
\newcommand{\mythemec}[1]{{\color[RGB]{51,108,121} #1}}
\newcommand{\mydarkred}[1]{{\color[rgb]{.6,.1,.1} #1}}
\newcommand{\mydarkblue}[1]{{\color[rgb]{.1,.1,.9} #1}}
\newcommand{\mygreenback}[1]{{\color[rgb]{.75,1,.75} #1}}
\newcommand{\myorange}[1]{{\color[rgb]{.76,.39,.13} #1}}
\newcommand{\mygrass}[1]{{\color[rgb]{.19,.64,.13} #1}}
\newcommand{\mysierp}[1]{{\color[RGB]{209,28,209} #1}}

\maketitle

\begin{abstract}
We develop a high-order hybridized discontinuous Galerkin (HDG) method for  linear degenerate elliptic equation arising from two-phase mixture of mantle convection or glacier dynamics.
We show that the proposed HDG method is well-posed by using an energy approach.
We carry out the {\em a priori} error estimates for the proposed HDG
  method on simplicial meshes in both two- and three-dimensions.
The error analysis shows that the convergence rates
are optimal for both the scaled pressure and the scaled velocity for non-degenerate problems and 
are sub-optimal by half order for degenerate ones.
 Several numerical results are presented to confirm the theoretical
 estimates. Degenerate problems with low regularity solutions are also studied and  numerical results show that  high-order methods are beneficial in terms of
 accuracy.
\end{abstract}




\pagestyle{myheadings} \thispagestyle{plain}
\markboth{S. Kang, T. Bui-Thanh, and T. Arbogast}{HDG methods for degenerate elliptic PDEs}

\section{Introduction}


The Earth's core is hotter than the Earth's surface, which leads to
thermal convection in which the cold mantle is dense and sinks while
the hot mantle is light and rises to the surface.  The induced
current, i.e., mantle convection, moves slowly and cools gradually.
The evolution and circulation of the mantle induce plate
tectonics, volcanic activity, and variation in crustal chemical composition.
Therefore, the study of mantle dynamics is critical to understanding
how the planet functions \cite{herring2010geodesy}.  Glacier dynamics,
on the other hand, describes the movement of glaciers and ice
sheets. Glaciers and ice sheets interact with the atmosphere, the oceans,
and the landscape \cite{chorley1971physical}, which could lead to a large
impact on weather and climate change \cite{kaser2001glacier}. Though
mantle convection and glacier dynamics are different in nature, their dynamics
can be mathematically modeled by the Stokes equations combined with a Darcy equation accounting for melt.


In this paper, we are interested in developing numerical methods for both glacial dynamics
and mantle convection 
 described by a similar two-phase mathematical model, as we now
briefly discuss.  In glacial dynamics, the mixture of ice and water is
observed near the temperature at the pressure-melting point, which is
in a phase-transition process
\cite{fowler1984transport,aschwanden2012enthalpy}.  In mantle
dynamics, a partially molten rock is generated by supplying heat or
reducing the pressure. In both cases, the relative motion between the melt and the
solid matrix is modeled by two-phase flow
\cite{mckenzie1984generation}.

We adopt the mathematical model in \cite{mckenzie1974convection,mckenzie1984generation,scott1984magma,scott1986magma,bercovici2001two,bercovici2003energetics,vsramek2007simultaneous}.
In particular, the mixture parameter of fluid melt and solid matrix is described by
the porosity $\phi$\textemdash the relative volume of fluid melt with respect to  the
bulk volume\textemdash which separates the solid single-phase ($\phi=0$) and
fluid-solid two-phase ($\phi>0$) regions \cite{arbogast2017mixed}.
The partially molten regions are governed by Darcy flow through
a deformable solid matrix which is modeled as a highly
viscous Stokes fluid
\cite{hewitt2008partial,vsramek2007simultaneous}. We use the subscript
$f$ and $s$ to distinguish between the fluid melt and solid matrix, and 
boldface lowercase letters for vector-valued functions.  We denote by $\vb_f$ and $\vb_s$ the velocities of fluid and solid, $\pres_f$ and $\pres_s$
the pressures, $\rho_f$ and $\rho_s$ the densities, $\mu_f$ and
$\mu_s$ the viscosities, and  $\sigma_f$ and $\sigma_s$  the
stresses.  
Darcy's law \cite{darcy1856fontaines,mckenzie1984generation} states
\begin{align}
\eqnlab{Darcy}
  \phi \LRp{\vb_f - \vb_s} &= 
    -\frac{\kappa(\phi)}{\mu_f}
  \LRp{\Grad \pres_f - \rho_f \gb},
\end{align}
where
$\kappa(\phi)$ is the permeability with $\kappa(0)=0$ and 
$\gb$ is the gravity. 
We assume that the solid matrix is more viscous than the fluid melt
 $(\mu_f \ll \mu_s)$ so that the fluid and the solid stresses can be modeled as
\begin{align}
  \sigma_f: &= -\pres_f \mc{I},\\
  \sigma_s: &= -\pres_s \mc{I} + \mu_s \LRp{\Grad \vb_s + \Grad \vb_s^T } 
               -\frac{2}{3} \mu_s \Div \vb_s \mc{I},
\end{align}
where $\mc{I}$ is the second order identity tensor.
The mixture of the melt and the solid matrix
obeys the Stokes equation \cite{bercovici2001two,vsramek2007simultaneous}
\begin{align}
\eqnlab{Stokes}
  \Div \LRp{\phi \sigma_f + (1-\phi) \sigma_s } 
  +\LRp{\phi \rho_f + (1-\phi) \rho_s } \gb = 0.
\end{align}

The mass conservations of the fluid melt and the solid matrix  are given as \cite{mckenzie1984generation,bercovici2001two}
\begin{align}
\eqnlab{MassConservation1}
  \dd{\LRp{\rho_f \phi}}{t} + \Div \LRp{\rho_f \phi \vb_f } & = 0,\\
\eqnlab{MassConservation2}
  \dd{\LRp{\rho_s (1-\phi)}}{t} + \Div \LRp{\rho_s (1-\phi) \vb_s } & = 0.
\end{align}
Applying a Boussinesq approximation (constant and equal densities for
non-buoyancy terms) to \eqnref{MassConservation1}--\eqnref{MassConservation2}, the total mass
conservation of the mixture can be written as
\begin{align}
\eqnlab{MassConservation}
  \Div \LRp{\phi \vb_f + (1-\phi) \vb_s} = 0.
\end{align}
The pressure jump between the melt and the matrix phases (the compaction relation) is given by \cite{sleep1988,bercovici2003energetics}
\begin{align}
\eqnlab{Compaction}
 \LRp{ \pres_s - \pres_f } = - \frac{\mu_s}{\phi} \Div \vb_s.
\end{align}

  The coupled Darcy-Stokes system \eqnref{Darcy}, \eqnref{Stokes},
  \eqnref{MassConservation} and \eqnref{Compaction} describes the
  motion of the mantle flow (and glacier dynamics). 
 The challenge is when $\phi=0$. Since solid matrix always exists, the Stokes part is well-posed, but the Darcy part is degenerate when $\phi=0$.

In this paper, we shall focus  on addressing the challenge of solving the linear degenerate elliptic equation
  arising from the Darcy part of the system. With a change of
  variables, \eqnref{Darcy} and a combination of \eqnref{MassConservation}--\eqnref{Compaction} 
  become
\begin{subequations}
  \eqnlab{DegenerateElliptic}
  \begin{IEEEeqnarray}{rCl}
    \tilde{\vb} + d(\phi)^2 (\Grad \pres - \tilde{\gb}) &= 0,  &\quad \text{ in } \Omega, \\
    \Div \tilde{\vb} + \phi \pres &= \phi^{\frac{1}{2}} f, &\quad \text{ in } \Omega, \\ 
    \phi \pres &= \phihalf g_D, & \quad\text{ on } \Gamma_D,
  \end{IEEEeqnarray}
\end{subequations}
where
$\Omega \subset \R^{dim}, dim=2\text{ or }3$, is an open and bounded domain,
$\Gamma_D=\pOmega$ the Dirichlet boundary,
$g_D$ the Dirichlet data,
$\nb$ the outward unit normal vector,
$\tilde{\vb}=\phi(\vb_f - \vb_s)$ the  Darcy velocity,
  $\pres=\pres_f$ the fluid pressure, $\tilde{\gb}=\rho_f \gb$,
 $f=\phi^{\frac{1}{2}}\pres_s$, $\mu_s=1$ and 
 $d(\phi) = \sqrt{\frac{\kappa(\phi)}{\mu_f}}$.
Though $d$ is a function of $\phi$, we shall write $d$ instead of $d\LRp{\phi} $ for the simplicity of the exposition.





  The boundary value problem  \eqnref{DegenerateElliptic} has been studied in \cite{arbogast2016linear},
  where the scaled velocity and pressure
  were proposed 
in order to obtain well-posedness. For numerical implementation, 
a cell-centered finite difference method \cite{arbogast2017cell}
and a mixed finite element method \cite{arbogast2016linear} have been studied.
The results showed that the numerical schemes are stable and have an optimal convergence rate for smooth solutions. 
However,  these schemes are low order accurate approaches.

Meanwhile, the high-order discontinuous Galerkin (DG) method\textemdash originally developed \cite{ReedHill73,LeSaintRaviart74,
  JohnsonPitkaranta86} for the neutron transport
equation\textemdash has
been studied extensively for virtually all types of partial
differential equations (PDEs)
\cite{douglas1976interior,wheeler1978elliptic,arnold1982interior,
cockburn2000development,arnold2002unified,
liu2004discontinuous,bassi2005discontinuous,
fezoui2005convergence,
noels2008explicit,
tirupathi2015modeling,gassner2016split,wintermeyer2017entropy}. This is due to the
fact that DG combines advantages of finite volume and finite element
methods. As such, it is well-suited to problems with large gradients
including shocks and with complex geometries, and large-scale
simulations demanding parallel implementations. 
In particular, for numerical modeling of magma dynamics, the DG methods have been used to study the interaction between the fluid melt and the solid matrix \cite{tirupathi2015multilevel,tirupathi2015modeling}, and to include a porosity-dependent bulk viscosity and a solid upwelling effect \cite{schiemenz2011modeling}.
In spite of these advantages, DG methods for
steady state and/or time-dependent problems that require implicit
time-integrators are more expensive in comparison to other existing
numerical methods, since DG typically has many more (coupled)
unknowns.

As an effort to mitigate the computational expense associated with DG
methods, the hybridized (also known as hybridizable) discontinuous
Galerkin (HDG) methods are introduced for various types of PDEs
including Poisson-type equation \cite{CockburnGopalakrishnanLazarov09,
  CockburnGopalakrishnanSayas10, KirbySherwinCockburn12,
  NguyenPeraireCockburn09a, CockburnDongGuzmanEtAl09,
  EggerSchoberl10}, Stokes equation \cite{CockburnGopalakrishnan09,
  NguyenPeraireCockburn10}, Euler and Navier-Stokes equations, wave
equations \cite{NguyenPeraireCockburn11, MoroNguyenPeraire11,
  NguyenPeraireCockburn11b, LiLanteriPerrrussel13,
  NguyenPeraireCockburn11a, GriesmaierMonk11, CuiZhang14, rhebergen2017hybridizable}, to name a
few. In \cite{Bui-Thanh15, Bui-Thanh15a, bui2016construction}, 
one of the authors has proposed an upwind HDG framework that provides a unified and 
systematic construction of HDG methods for a large class of PDEs.
We note that the weak Galerkin (WG) methods
in \cite{wang2013weak, WangYe14, zhai2015hybridized,mu2014new}
share many similar advantages with HDG. In fact, HDG and WG are the same for the degenerate elliptic problem in this paper.


The main goal of this paper is to develop a high-order HDG scheme for the linear degenerate elliptic equation 
\eqnref{DegenerateElliptic}.
In section \secref{ScaledSystem},
 we briefly discuss the scaled system for \eqnref{DegenerateElliptic}. 
In section \secref{HDGformulation},
 we derive the HDG formulation for the scaled system based on 
 the upwind HDG framework.
 The key feature is that we have modified the upwind HDG flux to accommodate the degenerate regions. 
 When the porosity vanishes, the resulting HDG system becomes ill-posed
 because the upwind parameter associated with the HDG flux disappears. 
 To overcome the difficulty, we introduce a generalized stabilization parameter
 that is an extension of the upwind based stabilization parameter. It has positive values on the degenerate interfaces.
 Next, we show the well-posedness and error analysis of the HDG system 
 under the assumption that the grid well matches with the intersection between the fluid melt and the solid matrix.
 In section \secref{NumericalResults},
 various numerical results for the scaled system will be presented to confirm the accuracy and robustness 
 of the proposed HDG scheme. Finally, we
 conclude the paper and discuss future research directions in section \secref{Conclusion}.



\section{Handling the degeneracy}
\seclab{ScaledSystem}


Let $\LRp{\cdot,\cdot}_\Omega$ be the $L^2$ inner-product on $\Omega$,
and $\LRa{\cdot,\cdot}_{\pOmega}$ be the $L^2$ inner-product on
$\pOmega$.  We denote the $L^2$ norm by
$\norm{\cdot}_\Omega=\LRp{\cdot,\cdot}_\Omega^\half$ on $\Omega$ and
by $\norm{\cdot}_{\pOmega} = \LRa{\cdot,\cdot}_{\pOmega}^\half$ on
$\pOmega$.  We also define the weighted $L^2$ norm on $\pOmega$ by
$\norm{\cdot}_{\pOmega,\tau} = \LRa{|\tau|
  \cdot,\cdot}_{\pOmega}^\half = \LRp{\int_{\pOmega} |\tau| (\cdot)^2
  dx}^\half$. For any $s \ne 0$, we denote the $H^s\LRp{D}$-norm as $\nor{\cdot}_{s,D}$, for example, $\norm{\cdot}_{\half,\pOmega}$ is the norm of $H^\half\LRp{\pOmega}$.

\subsection{The scaled system}
When the porosity becomes zero, the system \eqnref{DegenerateElliptic} degenerates.
However, we can still investigate how the solutions behave as the porosity vanishes. 
According to \cite{arbogast2016linear}, 
a priori energy estimates for the system \eqnref{DegenerateElliptic} read as
\begin{align}
   \norm{ \d^{-1}\tilde{\vb} }_\Omega
   + \norm{ \phihalf \pres}_\Omega
   + \norm{\phimhalf \Div \tilde{\vb} }_\Omega
   & \le
   c
   \LRp{
     \norm{g_D}_{H^\half(\pOmega)}
    + \norm{\d \tilde{\gb} }_\Omega
    + \norm{f}_\Omega
   },
\end{align}
for some constant $c > 0$.
Note that we may lose control of the pressure $\pres$ 
as the porosity approaches zero. 
This is because 
the fluid pressure $\pres$ is not defined in the solid regions.

To have the control of the pressure, 
following \cite{arbogast2016linear}, we define the scaled velocity and the scaled pressure as
$\ub = \d^{-1} \tilde{\vb}$ and $\p = \phi^{\frac{1}{2}} \pres$, respectively. 
The system \eqnref{DegenerateElliptic} becomes
\begin{subequations}
  \eqnlab{ScaledDegenerateElliptic}
  \begin{IEEEeqnarray}{rCl}
    \ub + \d \Grad \LRp{ \phimhalf \p } &= \d \tilde{\gb}, & \quad \text{ in } \Omega,\\
    \phimhalf \Div \LRp{\d \ub } +  \p &= f,  & \quad \text{ in } \Omega,\\
    \eqnlab{ScaledDegenerateEllipticBC}
    \p  &= g_D, & \quad \text{ on } \Gamma_D.
   \end{IEEEeqnarray}
\end{subequations}
Here, we interpret the differential operators in \eqnref{ScaledDegenerateElliptic} as
\begin{align}
  \d \Grad (\phimhalf \p) 
    &= -\frac{1}{2} \phi^{-\frac{3}{2}} \d \Grad \phi \p 
    + \phimhalf \d \Grad \p,\\
  \phimhalf \Div (\d \ub) 
    &= \phimhalf \Grad \d \cdot \ub + \phimhalf \d \Div \ub,
\end{align}
where we assume that 
\begin{subequations}
  \eqnlab{porosityCoefAssumption}
\begin{align}
  \phi^{-\frac{1}{2}} \d & \in L^{\infty}(\Omega),\\
  \phi^{-\frac{1}{2}} \Grad \d & \in (L^{\infty}(\Omega))^{dim},\\
  \phi^{-\frac{3}{2}} \d \Grad \phi & \in (L^{\infty}(\Omega))^{dim}.
\end{align}
\end{subequations}
With the assumption \eqnref{porosityCoefAssumption}, 
the scaled system does not degenerate. 
If the porosity vanishes, then $d(\phi)=0$, which leads to $\ub=0$ and $\p=f$.
The energy estimates for the scaled system \eqnref{ScaledDegenerateElliptic} read as
\begin{align}
   \norm{ u}_\Omega
   + \norm{ \p}_\Omega
   + \norm{\phimhalf \Div \LRp{ \d \ub }  }_\Omega
   & \le
   c
   \LRp{
     \norm{g_D}_{H^\half(\pOmega)}
    + \norm{\d \tilde{\gb} }_\Omega
    + \norm{f}_\Omega
   },
\end{align}
for some constant $c > 0$ \cite{arbogast2016linear}.  
We clearly see that we have control of the scaled pressure $\p$ even when the porosity becomes zero.


\subsection{Upwind-based HDG flux}
With some simple manipulation,
 the scaled system \eqnref{ScaledDegenerateElliptic} can be rewritten as 
\begin{subequations}
  \eqnlab{ScaledDegenerateElliptic2}
  \begin{IEEEeqnarray}{rCl}
    \ub
    - \phimhalf \Grad \d \p
    + \Div \LRp{\phimhalf \d \p \mc{I} } 
    &= \d \tilde{\gb}, & \quad \text{ in } \Omega,\\
    \frac{1}{2} \phi^{-\frac{3}{2}} \d \Grad \phi \cdot \ub 
    + \p
    + \Div \LRp{\phimhalf \d \ub }
    &= f,  & \quad \text{ in } \Omega.
  \end{IEEEeqnarray}
\end{subequations}
We cast the scaled system \eqnref{ScaledDegenerateElliptic2} 
into the conservative form
\begin{equation}
\eqnlab{LPDE}
\Div \F(\rb) + \mc{G} \rb = \fb, \quad \text{ in } \Omega,
\end{equation}
where we have defined 
the scaled velocity $\ub:=\LRp{\u_1,\u_2,\u_3}$,  
the solution vector $\rb := \LRp{\u_1,\u_2,\u_3, \p}$,
the source vector $\fb := \LRp{\d\tilde{g}_1, \d \tilde{g}_2, \d \tilde{g}_3 ,\f}$,
the flux tensor 
\begin{align}
  \F := 
  \LRp{F_1, F_2, F_3} := \F\LRp{\rb} :=  
  \phimhalf \d
  \begin{pmatrix}
    \p & 0 & 0 \\
    0 & \p & 0 \\
    0 & 0 & \p \\
    u_1 & u_2 & u_3
  \end{pmatrix}
\end{align} and 
\begin{align}
  \mc{G} := 
  \begin{pmatrix}
    1 & 0 & 0 & -\phi^{-\half}\pp{ \d }{x} \\
    0 & 1 & 0 & -\phi^{-\half}\pp{ \d }{y} \\
    0 & 0 & 1 & -\phi^{-\half}\pp{ \d }{z} \\
    \frac{1}{2}\phi^{-\frac{3}{2}}   \d \dd{\phi}{x}
    & \frac{1}{2}\phi^{-\frac{3}{2}} \d \dd{\phi}{y}
    & \frac{1}{2}\phi^{-\frac{3}{2}} \d \dd{\phi}{z}
    & 1 \\
  \end{pmatrix}.
\end{align}

We define the normal vector $\nb:= \LRp{n_1,n_2,n_3}$ and the flux Jacobian 
\begin{align}
  \mc{A} 
  = \sum_{k=1}^{3} n_k \dd{F_k}{\rb} 
  = \phimhalf \d 
  \begin{pmatrix}
    0 & 0 & 0 & n_1 \\
    0 & 0 & 0 & n_2 \\
    0 & 0 & 0 & n_3 \\
   n_1& n_2& n_3 & 0 
  \end{pmatrix}, 
\end{align}
which has four eigenvalues $\LRp{-\phimhalf \d, 0, 0, \phimhalf \d}$ and distinct eigenvectors 
\begin{align}
  W_1=
  \begin{pmatrix}
    -n_1\\
    -n_2\\
    -n_3\\
      1
  \end{pmatrix},
  W_2=
  \begin{pmatrix}
    -n_2\\
     n_1\\
      0\\
      0
  \end{pmatrix},
  W_3=
  \begin{pmatrix}
    -n_3\\
      0\\
     n_1\\
      0
  \end{pmatrix},  
  \text{ and }
  W_4=
  \begin{pmatrix}
    n_1\\
    n_2\\
    n_3\\
     1
  \end{pmatrix}.
\end{align}

 The system \eqnref{LPDE} can be considered as a steady state hyperbolic system \cite{toro2013riemann}. 
Finally, following the upwind HDG framework in \cite{Bui-Thanh15} we can construct
the upwind HDG flux with scalar $\ph$ and vector $\ubh$ trace unknowns as 
\begin{equation}
\eqnlab{HDGflux}
\Fh (\hat{\rb})\cdot \nb = \phimhalf \d
  \begin{pmatrix}
    n_1 \ph\\
    n_2 \ph\\
    n_3 \ph\\
    \ubh \cdot \nb
  \end{pmatrix}
  := \F (\rb) \cdot \nb + |\mc{A}|(\rb-\hat{\rb}) 
  = \phimhalf \d 
  \begin{pmatrix}
    n_1 \LRp{\p + (\ub-\ubh)\cdot \nb} \\
    n_2 \LRp{\p + (\ub-\ubh)\cdot \nb} \\
    n_3 \LRp{\p + (\ub-\ubh)\cdot \nb} \\
    \ub \cdot \nb + (\p-\ph)
  \end{pmatrix},
\end{equation}
where $\hat{\rb}=\LRp{\hat{u}_1,\hat{u}_2,\hat{u}_3,\ph}$,
$|\mc{A}|:=W|D|W^{-1}$, $D$ is the diagonal matrix of eigenvalues of $W_1$, $W_2$, $W_3$ and $W_4$, and $W$
is the matrix of corresponding eigenvectors. 
Following \cite{Bui-Thanh15}, we can compute $\ubh\cdot\n$ as a function of $\ph$, and hence $\ubh\cdot\n$ can be eliminated.
The upwind HDG flux can be then written in terms of $\ph$ as
\begin{align}
  \Fh (\hat{\rb})\cdot \nb = \phimhalf \d
  \begin{pmatrix}
    n_1 \ph\\
    n_2 \ph\\
    n_3 \ph\\
    \ub \cdot \nb + (\p-\ph)
  \end{pmatrix}.
\end{align}

\section{HDG formulation}
\seclab{HDGformulation}
We denote by
$\Omega_h := \cup_{i=1}^\Ne \K_i$ the mesh containing a finite
collection of non-overlapping elements, $\K_i$, that partition
$\Omega$.  Here, $h$ is defined as $h := \max_{j\in
  \LRc{1,\hdots,\Ne}}diam\LRp{\Kj}$. Let $\pOmega_h := \LRc{\pK:\K
  \in \Omega_h}$ be the collection of the boundaries of all elements. Let
us define $\Gh:= \Gho \cup \Ghb$ as the skeleton of
the mesh which consists of the set of all uniquely defined faces/interfaces,
where $\Ghb$ is the set of all boundary faces on $\pOmega$, and
$\Gho=\Gh \setminus \Ghb$ is the set of all interior interfaces. For
two neighboring elements $\Kp$ and $\Km$ that share an interior
interface $\e = \Kp \cap \Km$, we denote by $q^\pm$ the trace of the
solutions on $\e$ from $\K^\pm$.  We define $\nm$ as the unit outward normal vector on
the boundary $\pK^-$ of element $\Km$, and $\np = -\nm$ the unit outward
normal of a neighboring element $\Kp$.  On the interior interfaces $\e
\in \Gho$, we define the mean/average operator $\averageM{\bf v}$,
 where $\vb $ is either a scalar or a vector quantify, as
$\averageM{{\bf v}}:=\LRp{{\bf v}^- + {\bf v}^+}/2$,
and the jump operator $\jump{\vb\cdot \nb}:= \vb^+\cdot \nb^+ + \vb^-\cdot \nb^-$ 
 On the boundary faces $\e \in \Ghb$, we define the mean and jump operators as
$\averageM{{\bf v}}:={\bf v}, \quad \jump{{\bf v}} :={\bf v}$.

Let ${\Poly^k} \LRp{D}$ denote the space of polynomials of degree at
most $k$ on a domain $D$. Next, we introduce discontinuous
piecewise polynomial spaces for scalars and vectors as
\begin{align*}
\Vh\LRp{\Omega_h} &:= \LRc{v \in L^2\LRp{\Omega_h}:
  \eval{v}_{\K} \in \Poly^k\LRp{\K}, \forall \K \in \Omega_h}, \\
\Lamh\LRp{\Gh} &:= \LRc{\lambda \in \Lte:
  \eval{\lambda}_{\e} \in \Poly^k\LRp{\e}, \forall \e \in \Gh},\\
\Vbh\LRp{\Omega_h} &:= \LRc{{\bf v} \in \LRs{L^2\LRp{\Omega_h}}^m:
  \eval{{\bf v}}_{\K} \in \LRs{\Poly^k\LRp{\K}}^m, \forall \K \in \Omega_h},\\
\Lambh\LRp{\Gh} &:= \LRc{\lambdab \in \LRs{\Lte}^m:
  \eval{\lambdab}_{\e} \in \LRs{\Poly^k\LRp{\e}}^m, \forall \e \in \Gh}.
\end{align*}
and similar spaces $\VhK$, $\Lamhe$, $\VbhK$, and $\Lambhe$ by
replacing $\Omega_h$ with $\K$ and $\Gh$ with $\e$. Here, $m$ is the
number of components of the vector under consideration.

We define the broken inner products as
$\LRp{\cdot,\cdot}_{\Omega_h} :=
\sum_{\K\in \Omega_h}\LRp{\cdot,\cdot}_\K$ and
$\LRa{\cdot,\cdot}_{\pOmega_h} :=
\sum_{\pK\in \pOmega_h}\LRa{\cdot,\cdot}_\pK$, and on the mesh
skeleton as $\LRa{\cdot,\cdot}_\Gh := \sum_{\e\in
  \Gh}\LRa{\cdot,\cdot}_\e$.
We also define the associated norms as
$\norm{\cdot}_\Omegah:= \LRp{ \sum_{K\in \Omegah} \norm{\cdot}_K^2 }^\half$, 
$\norm{\cdot}_\pOmegah:= \LRp{ \sum_{K\in \Omegah} \norm{\cdot}_{\pK}^2 }^\half$, and the weighted norm
$\norm{\cdot}_{\pOmegah,\tau}:= \LRp{ \sum_{K\in \Omegah} \norm{\cdot}_{\pK,\tau}^2 }^\half$ 
 (recall $\norm{\cdot}_{\pK,\tau}=\norm{|\tau|^\half (\cdot)}_{\pK}$).

\subsection{Weak form}
From now on, 
{\it we conventionally use $\ub^e$, $\p^e$ and $\ph^e$ for the exact solution while $\ub$, $\p$ and $\ph$ are used to denote the HDG solution.} 
Unlike the DG approach, in which $\ph$ on an interface is computed using
information from neighboring elements that share that interface, i.e.,
\begin{equation}
  \eqnlab{upwindQh}
  \ph = \half \jump{\ub\cdot\n} + \averageM{\p},
\end{equation}
the idea behind HDG is to treat $\ph$ as a new unknown.  Testing
\eqnref{ScaledDegenerateElliptic2} or \eqnref{LPDE} with 
$(\testv,\test)$
and integrating
by parts we obtain the local solver for each element by replacing the flux $\LRa{\F\cdot\nb,(\vb,\q)}_\pK$ with the HDG numerical flux $\LRa{\Fh\cdot\nb,(\vb,\q)}_\pK$. The local solver reads: find
$\LRp{\ub,\p,\ph}  \in \Vbh(\K) \times \Vh(\K) \times \Lamh(\pK)$
 such that
\begin{subequations}
  \eqnlab{LocalSolver}
  \begin{align}
    \LRp{\ub,\testv}_\K
    - \LRp{\phimhalf \Grad \d  \p, \testv}_\K
    - \LRp{ \phimhalf \d \p ,\Div \testv }_\K
    + \LRa{\phimhalf \d \ph, \testv \cdot \nb}_\pK 
    = \LRp{ \d \tilde{\gb},\testv}_\K,\\ 
    \LRp{\p,\test}_\K 
    + \LRp{\frac{1}{2} \phi^{-\frac{3}{2}} \d \Grad \phi \cdot \ub, \test }_\K
    - \LRp{\phimhalf \d  \ub,\Grad \test }_\K
    + \LRa{\phimhalf \d  \LRp{ \ub \cdot \nb + (\p - \ph)},\test}_\pK 
    = \LRp{f,\test}_K,
  \end{align}
\end{subequations}
for all $(\testv,\test) \in \Vbh(\K) \times \Vh(\K)$.

Clearly we need an additional equation to close the system since we
have introduced an additional trace unknown  $\ph$. The natural condition
is the conservation, that is, the continuity of the HDG flux. For the HDG method to be conservative, it is  sufficient to 
weakly enforce the continuity of the last component of the HDG flux
\eqnref{HDGflux} on each face $e$ of the mesh skeleton, i.e.,
\begin{equation}
\eqnlab{GlobalSolver}
  \LRa{\jump{\phimhalf \d \ub\cdot\nb + \phimhalf \d \LRp{\p-\ph}},\testmu}_\e = 0, \quad \forall \e \in \Gho.
\end{equation}

On degenerate faces, where $\phi = 0$, the conservation condition \eqnref{GlobalSolver} is
trivially satisfied. These faces would need to be sorted out and removed from the system. 
However, this creates implementation issues. 
To avoid this, we introduce a more general
HDG flux
\begin{equation}
\eqnlab{HDGfluxGeneral}
\Fh\cdot \nb := 
  \begin{pmatrix}
    n_1 \phimhalf \d \ph \\
    n_2 \phimhalf \d \ph \\
    n_3 \phimhalf \d \ph \\
    \phimhalf \d \ub \cdot \nb + \tau(\p-\ph)
  \end{pmatrix},
\end{equation}
where $\tau$ is a positive function on the edge. For example, we can take $\tau =
\phimhalf \d$ for non-degenerate faces (i.e., faces with $\phi > 0$),
and for degenerate ones (i.e. faces with $\phi=0$) we take
$\tau=\gamma > 0$. Alternatively, we can take a single value $\tau =
\mc{O}\LRp{1/\h}$ over the entire mesh skeleton. We shall compare these choices in Section \secref{NumericalResults}.
With this HDG flux, the conservation condition \eqnref{GlobalSolver} becomes
\begin{equation}
\eqnlab{GeneralGlobalSolver}
  \LRa{\jump{\phimhalf \d \ub\cdot\nb + \tau \LRp{\p-\ph}},\testmu}_\e = 0, \quad \forall \e \in \Gho, \quad \forall \testmu \in \Lamh(\e).
\end{equation}
On the Dirichlet boundary $\Gamma_D$,
 we impose the boundary data $g_D$ to $\ph$ through the weak form of 
\begin{equation}
\eqnlab{GlobalSolverBC}
  \LRa{\tau\ph,\testmu}_{\Gamma_D} = \LRa{\tau g_D,\testmu}_{\Gamma_D}, \quad \forall \testmu \in \Lamh(\Gamma_D).
\end{equation}
With the general HDG flux \eqnref{HDGfluxGeneral} and Dirichlet boundary condition \eqnref{ScaledDegenerateEllipticBC},
the local equation \eqnref{LocalSolver} 
 now becomes 
\begin{subequations}
  \eqnlab{GeneralLocalSolver}
  \begin{align}
    \LRp{\ub,\testv}_\K
    - \LRp{\phimhalf \Grad \d \p, \testv}_\K
    - \LRp{ \phimhalf \d \p ,\Div \testv }_\K
    + \LRa{\phimhalf \d \ph, \testv \cdot \nb}_{\pK \setminus \Gamma_D} \nonumber \\
    + \LRa{\phimhalf \d g_D, \testv \cdot \nb}_{\pK \cap \Gamma_D}
    = \LRp{\d \tilde{\gb},\testv}_\K,\\ 
    \LRp{\p,\test}_\K 
    + \LRp{\frac{1}{2} \phi^{-\frac{3}{2}} \d  \Grad \phi \cdot \ub, \test }_\K
    - \LRp{\phimhalf \d  \ub,\Grad \test }_\K
    + \LRa{\phimhalf \d  \ub \cdot \nb + \tau(\p - \ph),\test}_{\pK \setminus \Gamma_D} \nonumber \\
    + \LRa{\phimhalf \d  \ub \cdot \nb + \tau(\p - g_D),\test}_{\pK \cap \Gamma_D}
    = \LRp{f,\test}_K.
  \end{align}
\end{subequations}

The HDG comprises the local solver \eqnref{GeneralLocalSolver}, the global equation
\eqnref{GeneralGlobalSolver} and the boundary condition \eqnref{GlobalSolverBC}. By summing \eqnref{GeneralLocalSolver} over all elements and \eqnref{GeneralGlobalSolver} over the mesh skeleton, we obtain the complete HDG system
with the weakly imposed Dirichlet boundary condition \eqnref{GlobalSolverBC}:
find $\LRp{\ub,\p,\ph} \in \Vbh(\Omegah) \times \Vh(\Omegah) \times \Lamh(\Gh)$ such that 
\begin{subequations}
  \eqnlab{HDGsystem}
  \begin{align}
    \eqnlab{HDGlocal1}
    \LRp{\ub,\testv}_\Omegah
    - \LRp{\phimhalf \Grad \d \p, \testv}_\Omegah
    - \LRp{ \phimhalf \d \p , \Div \testv }_\Omegah
    + \LRa{\phimhalf \d \ph, \testv \cdot \nb}_{\pOmegah \setminus \Gamma_D} \nonumber\\
    = \LRp{ \d \tilde{\gb},\testv}_\Omegah 
    - \LRa{g_D,\phimhalf \d  \testv \cdot \nb}_{\Gamma_D},\\
    \eqnlab{HDGlocal2}
    \LRp{\p,\test}_\Omegah 
    + \LRp{\frac{1}{2} \phi^{-\frac{3}{2}} \d  \Grad \phi \cdot \ub, \test }_\Omegah
    - \LRp{\phimhalf \d  \ub,\Grad \test }_\Omegah
    + \LRa{\phimhalf \d  \ub\cdot\nb + \tau \LRp{\p-\ph},\test }_{\pOmegah \setminus \Gamma_D} \nonumber \\
    + \LRa{\phimhalf \d  \ub\cdot\nb + \tau \p,\test }_{ \Gamma_D} 
    = \LRp{f,\test}_\Omegah
    + \LRa{ \tau  g_D,\test }_{\Gamma_D},\\
    \eqnlab{HDGglobal}
    - \LRa{\jump{\phimhalf \d \ub\cdot\nb + \tau \LRp{\p-\ph}},\testmu}
    _{\Gh \setminus \Gamma_D}
    + \LRa{\tau \ph,\testmu}_{\Gamma_D}
    = \LRa{ \tau g_D,\testmu}_{\Gamma_D},
  \end{align}
\end{subequations}
for all $(\testv,\test,\testmu) \in \Vbh(\Omegah) \times \Vh(\Omegah)
\times \Lamh(\Gh)$. Note that this form resembles the weak Galerkin
framework 
\cite{wang2013weak, WangYe14, zhai2015hybridized,mu2014new}.
Indeed, HDG and the weak Galerkin method are equivalent in this case.

The HDG computation consists of three steps: first,
solve the local solver for $\LRp{\ub,\p}$ as a function of $\ph$
element-by-element, completely independent of each other; second,
substitute $\LRp{\ub,\p}$ into the global equation \eqnref{GeneralGlobalSolver} to
solve for $\ph$ on the mesh skeleton; and finally recover the local
volume unknown $\LRp{\ub,\p}$ in parallel.

\subsection{Well-posedness}
Let us denote the bilinear form on the left hand side of
\eqnref{HDGsystem} as $a\LRp{\LRp{\ub,\p, \ph}; \LRp{\vb,\q,\mu}}$ and the linear form on right hand side as
$\ell\LRp{\LRp{\vb,\q,\mu}}$. We begin with an energy estimate for the HDG solution.
\begin{proposition}[Discrete energy estimate]
  \propolab{energyEstimate}
Suppose $\gD \in {L^2(\Gamma_D)}$, $\f \in L^2\LRp{\Omegah}$, and
$d(\phi) \tilde{\gb} \in L^2\LRp{\Omegah}$.
If $\tau=\mc{O}(1/h)$, then it holds that 
\begin{align}
a\LRp{\LRp{\ub,\p, \ph}; \LRp{\ub,\p, \ph}} 
      &= \nor{\ub}_\Omegah^2 + \nor{\p}_\Omegah^2 + \nor{\ph}^2_{\gamD,\tau} + \nor{\p}^2_{\gamD,\tau} + \nor{\p - \ph}_{\pOmegah\setminus\gamD,\tau}^2 \\
      & \le c\LRp{\nor{\gD}^2_{\gamD,\tau} + \nor{\tilde{\gb}}^2_\Omegah + \nor{\f}^2_\Omegah},
\end{align}
for some positive constant $c = c\LRp{\phi,d,\tau, h, k}$. In
particular, there is a unique solution $\LRp{\ub,\p,\ph}$ to the HDG system \eqnref{HDGsystem}.
\end{proposition}
\begin{proof}
  We start with the following identities
    \begin{subequations}
    \eqnlab{identities}
    \begin{align}
      \eqnlab{identitiesU}
      - \LRp{ \phimhalf \Grad \d \p, \vb}_\K
      - \LRp{ \phimhalf \d \p ,\Div \vb }_\K = 
      - \LRp{\p,\phimhalf \Div \LRp{\d \vb } }_\K,\\
      \eqnlab{identitiesP}
      \LRp{\frac{1}{2}\phi^{-\frac{3}{2}} \d \Grad \phi \cdot \ub,q}_\K 
      - \LRp{ \phimhalf \d \ub, \Grad q}_K
      = 
      \LRp{ \phimhalf \Div \LRp{\d \ub }, q }_\K
      - \LRa{ \phimhalf \d \ub\cdot \nb, q}_\pK.
    \end{align}
    \end{subequations}

Now taking $\vb = \ub$, $\q = \p$, and $\qh = \ph$ in
\eqnref{HDGsystem} and \eqnref{identities}, and then adding all equations in \eqnref{HDGsystem} gives
\begin{align*}
  a\LRp{\LRp{\ub,\p, \ph}; \LRp{\ub,\p, \ph}} &= \nor{\ub}_\Omegah^2 +
  \nor{\p}_\Omegah^2 + \nor{\ph}^2_{\gamD,\tau} + \nor{\p}^2_{\gamD,\tau} + \nor{\p -
    \ph}_{\pOmegah\setminus\gamD,\tau}^2 = \\
  &-\LRa{\gD, \phimhalf
    d\ub\cdot\n}_\gamD - \LRa{\tau\gD,\p}_\gamD + \LRa{\tau\gD,\ph}_\gamD +
  \LRp{d\tilde{\gb},\ub}_\Omegah + \LRp{\f,\p}_\Omegah,
\end{align*}
which, after invoking the Cauchy-Schwarz and Young inequalities, becomes
\begin{multline*}
a\LRp{\LRp{\ub,\p, \ph}; \LRp{\ub,\p, \ph}} \le
\frac{\nor{\phimhalf d}_\infty}{2\varepsilon_1} \nor{\gD}_{\gamD,\tau}^2 +  \frac{\varepsilon_1}{2}
\nor{\ub}_{\gamD,\tau^{-1}}^2 \\+ \frac{1}{2\varepsilon_2} \nor{\gD}_{\gamD,\tau}^2 + \frac{\varepsilon_2}{2} \nor{\p-\ph}_{\gamD,\tau}^2
+\frac{1}{2\varepsilon_3} \nor{d\tilde{\gb}}_{\Omegah}^2+ \frac{\varepsilon_3}{2} \nor{\ub}_{\Omegah}^2 +
\frac{1}{2\varepsilon_4} \nor{\f}_{\Omegah}^2+ \frac{\varepsilon_4}{2} \nor{\p}_{\Omegah}^2,
\end{multline*}
which yields the desired energy estimate after applying an inverse trace inequality (c.f. Lemma \eqnref{lemmaA1}) for the second term on right hand side and choosing sufficiently small values for $\varepsilon_1, \varepsilon_2, \varepsilon_3$ and $ \varepsilon_4$.
\end{proof}

Since the HDG system \eqnref{HDGsystem} is linear and square in terms
of the HDG variables $\LRp{\ub,\p,\ph}$, the uniqueness result in
Proposition \proporef{energyEstimate} implies existence and stability, and
hence the well-posedness of the HDG system. 

  \begin{lemma}[Consistency]
    Suppose $\LRp{\ub^e,\p^e}$ is a weak solution of
    \eqnref{ScaledDegenerateElliptic}, which is sufficiently
    regular. Then $\LRp{\ub^e,\p^e, \eval{\p^e}_{\Gh}}$ satisfies the
    HDG formulation \eqnref{HDGsystem}. In particular, the Galerkin orthogonality holds, i.e.,
    \begin{equation}
     \eqnlab{GalerkinOrth}
    a\LRp{\LRp{\ub^e-\ub,\p^e - \p, \eval{\p^e}_{\Gh} - \ph}; \LRp{\vb,\q,\mu}} = 0, \quad \forall (\testv,\test,\testmu) \in \Vbh(\Omegah) \times \Vh(\Omegah) \times \Lamh(\Gh).
    \end{equation}
  \end{lemma}
    The proof is a simple application of integration by parts and hence omitted.

  \subsection{Error analysis}
\seclab{errorAnalysis}
We restrict the analyis for simplicial meshes and adopt the projection-based error analysis
  in \cite{CockburnGopalakrishnanSayas10}.  To begin,
  we define $\phe$ as the trace of $\pe$. For any element $\K$, $e \in \Gh$, $e \subset \pK$,
  we denote by $\P \LRp{\ube, \pe, \phe} := \LRp{\Prb\ube, \Pr\pe,
    \Pi\phe}$, where $\Pi$ is the standard $L^2$-projection, a
  collective projection of the exact solution.
  Let us define
  \begin{align}
    \eub^I &:= \ube - \Prb\ube, \quad \eub^h := \ub - \Prb\ube, \\
    \ep^I &:= \pe - \Pr\pe, \quad \ep^h := \p  - \Pr\pe, \\
    \eph^I &:= \phe - \Pi\phe, \quad \eph^h := \ph - \Pi\phe,
  \end{align}
  and then
  the projections
  $\Pr\ube$ and $\Pr\pe$ are defined by
   \begin{subequations}
    \eqnlab{projection}
    \begin{align}
      \eqnlab{proju}
      (\eub^I, \vb)_\K &= 0, \quad  \vb \in \LRs{\poly{k-1}(\K)}^{dim}, \\
      \eqnlab{projp}
      (\ep^I, \q)_\K &= 0, \quad \q \in \poly{k-1}(\K), \\
      \eqnlab{projup}
      \LRa{ {\alpha\eub^I \cdot \n + \tau \ep^I} , \qh }_e &= 0, \quad \qh \in \poly{k}(e),
      \end{align}
   \end{subequations}
for each $\K \in \Omegah$, $e \in \Gh$ and $e \subset \pK$.
Here
$\alpha$, to be defined below, is a  positive constant on each face $\e$ of element $\K$.
\begin{lemma}
\lemlab{eProjection}
  Let $\tau_\K := \tau/\alpha$.  The projections $\Prb\ube$ and $\Pr\pe$ are well-defined, and 
  \begin{align*}
    \nor{\eub^I }_{\K} + h\nor{\eub^I }_{1,\K} &\le ch^{k +1} \nor{\ube }_{k+1, \K} + c h^{k +1}\tau_\K^* \nor{\pe }_{k +1, \K} , \\
\nor{ \ep^I }_{\K} + h\nor{ \ep^I }_{1,\K}&\le c\frac{h^{k +1}}{\tau_\K^{\max}} \nor{ \Div \ube }_{k, \K} + c h^{k +1} \nor{\pe}_{k +1, \K},
  \end{align*}
    where $\tau_\K^{\max}:= \max\eval{\tau_\K}_{\pK}$ and $\tau_\K^*:= \eval{\tau_\K}_{\pK \setminus \e^*}$, where $\e^*$ is the edge on which $\tau_\K$ is maximum.
\end{lemma}
  The proof can be obtained from \cite{CockburnGopalakrishnanSayas10}.

Since the interpolation errors $\eub^I, \ep^I$ and $\eph^I$ have
optimal convergence order, by the triangle inequality, the convergent
rates of the total errors $\eub = \eub^I + \eub^h, \ep = \ep^I +
\ep^h$, and $\eph = \eph^I + \eph^h$ are determined by those of the
discretization errors $\eub^h, \ep^h$ and $\eph^h$. We use an
energy approach to estimate the discretization errors. To begin, let us define
\[
\mc{E}^2_h := \nor{\eub^h}_\Omegah^2 + \nor{\ep^h}_\Omegah^2 + \nor{\eph^h}^2_{\gamD,\tau} + \nor{\ep^h}^2_{\gamD,\tau} + \nor{\ep^h - \eph^h}_{\pOmegah\setminus\gamD,\tau}^2.
\]

\begin{lemma}[Error equation]
\lemlab{ErrorEqn}
  It holds that 
  \begin{multline}
  \mc{E}^2_h = \underbrace{-\LRp{\phimhalf\ep^I,\Div\LRp{\d\eub^h}}_\Omegah}_{A} 
  + \underbrace{\LRa{\phimhalf\d\eph^I,\eub^h\cdot\n}_{\pOmegah\setminus\gamD}}_{B} \underbrace{-\LRp{\d\eub^I,\Grad\LRp{\phimhalf\ep^h}}_\Omegah}_C \\+
  \underbrace{\LRa{\LRp{\phimhalf\d - \alpha}\eub^I\cdot\n,\ep^h-\eph^h}_{\pOmegah\setminus\gamD}
  +\LRa{\LRp{\phimhalf\d - \alpha}\eub^I\cdot\n,\ep^h}_{\gamD}}_D
  \end{multline}
  \end{lemma}
    \begin{proof}
      The proof is straightforward by first adding and subtracting
      appropriate projections in the Galerkin orthogonality equation
      \eqnref{GalerkinOrth}, second using the definition of the
      projections \eqnref{projection}, and finally taking $\vb = \eub^h$,
      $\q = \ep^h$, and $\qh = \eph^h$.
    \end{proof}

    The next step is to estimate $A, B, C$ and $D$. To that end, we define $\alpha$ on faces of an element $\K$ as
\begin{equation}
\eqnlab{alpha}
\alpha:= 
\left\{
\begin{array}{ll}
\overline{\phimhalf\d} & \text{if } \overline{\phimhalf\d} \ne 0\\
1                      & \text{otherwise}
\end{array}
\right.,
\end{equation}
where $\overline{\phimhalf\d}$ is the average of $\phimhalf\d$ on the element $\K$.

    \begin{lemma}[Estimation for $A$]
      \lemlab{Aestimate}
      There exists a positive constant $c=c\LRp{\phi,\d}$ such that
      \[
      \snor{A} \le c \nor{\ep^I}_\Omegah\nor{\eub^h}_\Omegah.
      \]
    \end{lemma}
    \begin{proof}
      We have
\[
        \snor{A} \le \snor{\LRp{\ep^I,\phimhalf\Grad\d\cdot\eub^h}_\Omegah} + \snor{\LRp{\ep^I,\phimhalf\d\Div{\eub^h}}_\Omegah} 
        \]
        Bounding the first term is straightforward:
        \[
        \snor{\LRp{\ep^I,\phimhalf\Grad\d\cdot \eub^h}_\Omegah}  \le       c \nor{\phimhalf\Grad\d}_\infty\nor{\ep^I}_\Omegah\nor{\eub^h}_\Omegah.
        \]

        For the second term, we have
        \begin{multline*}
          \snor{\LRp{\ep^I,\phimhalf\d\Div{\eub^h}}_\Omegah} =
          \snor{\LRp{\ep^I,\LRp{\phimhalf\d -
                \overline{\phimhalf\d}}\Div{\eub^h}}_\Omegah} \le ch\nor{\ep^I}_\Omegah\nor{\phimhalf\d}_{W^{1,\infty}\LRp{\Omegah}}\nor{\Div{\eub^h}}_\Omegah \\ \le c \nor{\ep^I}_\Omegah\nor{\phimhalf\d}_{W^{1,\infty}\LRp{\Omegah}}\nor{\eub^h}_\Omegah,
        \end{multline*}
        where we have used \eqnref{projp} in the first equality, the Cauchy-Schwarz inequality and the
        Bramble--Hilbert lemma (see, e.g., \cite{Brenner-Scott-book})
        in the first inequality, and Lemma \lemref{inverse-inequality} (in the appendix)
        in the last inequality. Here, $W^{1,\infty}$ is a standard
        Sobolev space.
      \end{proof}

        \begin{lemma}[Estimation for $B$]
      \lemlab{Bestimate}
      There exists a positive constant $c=c\LRp{\phi,\d}$ such that
      \[
      \snor{B} \le c h^\half\nor{\ep^I}_\pOmegah \nor{\eub^h}_\Omegah.
      \]
    \end{lemma}
    \begin{proof}
      We have
      \begin{multline*}
        \snor{B} = \snor{\LRa{\eph^I,\LRp{\phimhalf\d - \overline{\phimhalf\d}}\eub^h\cdot\n}_{\pOmegah\setminus\gamD}}
        \le \nor{\eph^I}_{\pOmegah}\nor{\phimhalf\d - \overline{\phimhalf\d}}_{L^\infty\LRp{\pOmegah}}\nor{\eub^h}_{\pOmegah} \\\le
        c \h\nor{\phimhalf\d}_{W^{1,\infty}\LRp{\Omegah}}\nor{\eph^I}_\pOmegah \nor{\eub^h}_\pOmegah,
        \end{multline*}
        where we have used the property of $L^2$-projection $\Pi\phe$
        in the first equality, the Cauchy-Schwarz inequality in the
        first inequality, and the Bramble--Hilbert lemma in the last inequality.
        Now the best approximation of $\Pi\phe$ implies
        $\norm{\eph^I }_\pOmegah \le \norm{\ep^I}_\pOmegah$ 
        and \eqnref{discrete_trace_ineq} gives the result.
      \end{proof}

    \begin{lemma}[Estimation for $C$]
      \lemlab{Cestimate}
      There exists a positive constant $c=c\LRp{\phi,\d}$ such that
      \[
      \snor{C} \le c \nor{\eub^I}_\Omegah\nor{\ep^h}_\Omegah.
      \]
    \end{lemma}
    \begin{proof}
      We have
\[
        \snor{C} \le  \snor{\LRp{\frac{1}{2}\phi^{-\frac{3}{2}} \d \Grad\phi \cdot \eub^I,\ep^h}_\Omegah} 
      +\snor{\LRp{ \eub^I, \phimhalf \d \Grad \ep^h}_\Omegah}.
        \]
        The rest of the proof is similar to that of Lemma \lemref{Aestimate} by using \eqnref{proju}.
      \end{proof}

        \begin{lemma}[Estimation for $D$]
      \lemlab{Destimate}
      There exists a positive constant $c=c\LRp{\phi,\d}$ such that
      \[
      \snor{D} \le c  \beta \nor{\eub^I}_{\pOmegah,\tau^{-1}} \LRp{\nor{\ep^h-\eph^h}_{\pOmegah\setminus\gamD,\tau} + \nor{\ep^h}_{\gamD,\tau}},
      \]
where
\[
\beta :=
\left\{
\begin{array}{ll}
h & \text{if } \overline{\phimhalf\d} \ne 0 \quad \forall \K \in \Omegah\\
1 & \text{otherwise} 
\end{array}
\right..
\]
    \end{lemma}
    \begin{proof}
Employing similar techniques as in estimating $B$, we have
\[
        \snor{D} \le 
        \nor{\phimhalf\d - \alpha}_{L^\infty\LRp{\pOmegah}} \nor{\eub^I}_{\pOmegah,\tau^{-1}} \LRp{\nor{\ep^h-\eph^h}_{\pOmegah\setminus\gamD,\tau} + \nor{\ep^h}_{\gamD,\tau}}.
\]
Now using the definition of $\alpha$ in \eqnref{alpha} and the Bramble--Hilbert lemma, 
\[
\nor{\phimhalf\d - \alpha}_{L^\infty\LRp{\pOmegah}} \le \nor{\phimhalf\d - \alpha}_{L^\infty\LRp{\Omegah}} \le c \beta,
\]
      and this ends the proof.
      \end{proof}

    Now comes the main result of this section.
    \begin{theorem}
      \theolab{convergence}
      Suppose $\ube \in \LRs{H^{k+1}\LRp{\Omegah}}^{dim}$ and  $\pe \in H^{k+1}\LRp{\Omegah}$. Then 
      \begin{multline*}
      \nor{\eub^h}_\Omegah + \nor{\ep^h}_\Omegah + \nor{\eph^h}_{\gamD,\tau} + \nor{\ep^h}_{\gamD,\tau} + \nor{\ep^h - \eph^h}_{\pOmegah\setminus\gamD,\tau} \\
\le c\LRp{\nor{\ube}_{k+1,\Omegah} + \nor{\pe}_{k+1,\Omegah}}
\times
\left\{
\begin{array}{ll}
 h^{k+1}  & \text{if } \overline{\phimhalf\d} \ne 0 \quad \forall \K \in \Omegah\\
h^{k+\half} & \text{otherwise} 
\end{array}
\right.,
      \end{multline*}
      where $c=c\LRp{\phi,\d,\tau}$ is a positive constant independent of $h$.
    \end{theorem}
\begin{proof}
Using the results in Lemmas \lemref{ErrorEqn}---\lemref{Destimate} and the Cauchy-Schwarz inequality, we have
\begin{multline}
\eqnlab{E1}
\mc{E}^2_h \le c\LRp{\nor{\ep^I}_\Omegah^2 + \beta \h \nor{\ep^I}_\pOmegah^2 + \nor{\eub^I}_\Omegah^2 + \beta \nor{\eub^I}_{\pOmegah,\tau^{-1}}^2}^\half \times\\
\LRp{\nor{\eub^h}_\Omegah^2 + \nor{\ep^h}_\Omegah^2 + \nor{\ep^h-\eph^h}_{\pOmegah\setminus\gamD,\tau}^2 + \nor{\ep^h}_{\gamD,\tau}^2}^\half \\
\le c\LRp{\nor{\ep^I}_\Omegah^2 + \beta \h\nor{\ep^I}_\pOmegah^2 + \nor{\eub^I}_\Omegah^2 + \beta  \nor{\eub^I}_{\pOmegah,\tau^{-1}}^2}^\half\times\mc{E}_h.
\end{multline}
The estimate for $\nor{\ep^I}_\Omegah^2$ and $\nor{\eub^I}_\Omegah^2$
can be obtained directly from Lemma \lemref{eProjection}. Now using Lemma \lemref{cont-inv-trace} in the appendix 
and approximation properties of $\Prb\ube, \Pr\pe$ in Lemma
\lemref{eProjection} gives
\begin{equation}
\eqnlab{E2}
\nor{\ep^I}_\pOmegah^2 \le c\sum_\K \LRp{\nor{\Grad \ep^I}_{0,\K} + h^{-1}\nor{\ep^I}_{0,\K} } \nor{\ep^I}_{0,\K} \le c h^{2k+1}\LRp{\max_\K\frac{1}{\tau_\K^{\max}}\nor{\ube}_{k+1,\Omegah} + \nor{\pe}_{k+1,\Omegah}}^2.
\end{equation}
Similarly we can obtain
\begin{equation}
\eqnlab{E3}
\nor{\eub^I}_{\pOmegah,\tau^{-1}}^2 \le ch^{2k+1}\max_\K\frac{1}{\tau} \LRp{\nor{\ube}_{k+1,\Omegah} + \max_\K\tau_\K^*\nor{\pe}_{k+1,\Omegah}}^2.
\end{equation}
The assertion is now ready by combining the inequalities \eqnref{E1}---\eqnref{E3}, the definition of $\beta$, and the Cauchy-Schwarz inequality.
\end{proof}

\begin{remark}
  When the system is degenerate, but the exact solution is piecewise
  smooth, the convergence rate is sub-optimal by half order. The above
  proof, especially inequality \eqnref{E3}, shows that this
  suboptimality may not be improved by using $\tau =
  \mc{O}\LRp{\h^{-1}}$. The reason is that the gain by half order from
  $\max_\K\frac{1}{\tau}$ is taken away by the loss of half order from
  $\max_\K\tau_\K^*$. This will be confirmed in our numerical studies of
  the sensitivity of $\tau$ on the convergence rate in Section \secref{sensitivityTau}.
  \end{remark}

\section{Numerical results}
\seclab{NumericalResults}

In this section, we present numerical examples to support the HDG approach and its convergence analysis.
For a non-degenerate case, we consider a sine solution test, while
for degenerate cases, we choose smooth and non-smooth solution tests \cite{arbogast2016linear}.
We take the upwind based parameter $\tau=\phimhalf d$ for a non-degenerate case, 
and the generalized parameter 
$$ \tau  = \left\{
\begin{array}{ll}
\phimhalf\d & \text{for }  \phi > 0,\\
1/\h      & \text{for } \phi=0
\end{array}
\right.$$
for degenerate cases. 
We also conduct several numerical computations to understand if the stabilization parameter $\tau$
can affect the accuracy of the HDG solution and its convergence rate. 
We assume that porosity $\phi$ is known and $\d =\phi$ in all the numerical examples.
The domain $\Omega$ is chosen as $\Omega=(0,1)^{dim}$ or  
$\Omega=(-1,1)^{dim}$, 
which is either uniformly discretized with $n_e$ rectangular tensor product elements in
each dimension (so that the total number of elements is $\Ne = n_e^{dim}$), 
or $N_e$ triangular elements.
Though we have rigorous optimal convergence theory for only simplicial meshes (see Theorem \theoref{convergence}), a similar result is expected for quadrilateral/hexahedral meshes (see the numerical results in the following sections). Since rectangular meshes are convenient for all problems in this paper with simple interfaces between the fluid melt and the solid matrix,
we use rectangular meshes hereafter, except for the test in Section \secref{simplexQuad}. 

\subsection{Non-degenerate case}
\seclab{simplexQuad}

We consider a non-degenerate case on $\Omega=(0,1)^2$ with the porosity given by $\phi = \exp(2(x+y))$. 
  We choose the pressure to be 
  $\pres^e = \exp(-(x+y)) \sin(m_x \pi x) \sin (m_y \pi y)$.
The corresponding manufactured scaled solutions are given as 
  \begin{subequations}
    \eqnlab{2d-tc3-solutions}
      \begin{align}
        \p^e &= \sin(m_x \pi x) \sin (m_y \pi y),\\
      u_x^e &= \exp(x+y) \sin (m_y \pi y) \LRp{\sin(m_x \pi x) - m_x \pi \cos(m_x\pi x)},\\
      u_y^e &= \exp(x+y) \sin (m_x \pi x) \LRp{\sin(m_y \pi y) - m_y \pi \cos(m_y\pi y)}.
      \end{align}
  \end{subequations}
  Here, we take $m_x = 2$ and $m_y = 3$. 
  \begin{figure}[h!t!b!]
    \centering
    \subfigure[Rectangular elements ]{
      \includegraphics[trim=5.2cm 9.0cm 6cm 9.1cm,clip=true,width=0.4\textwidth]{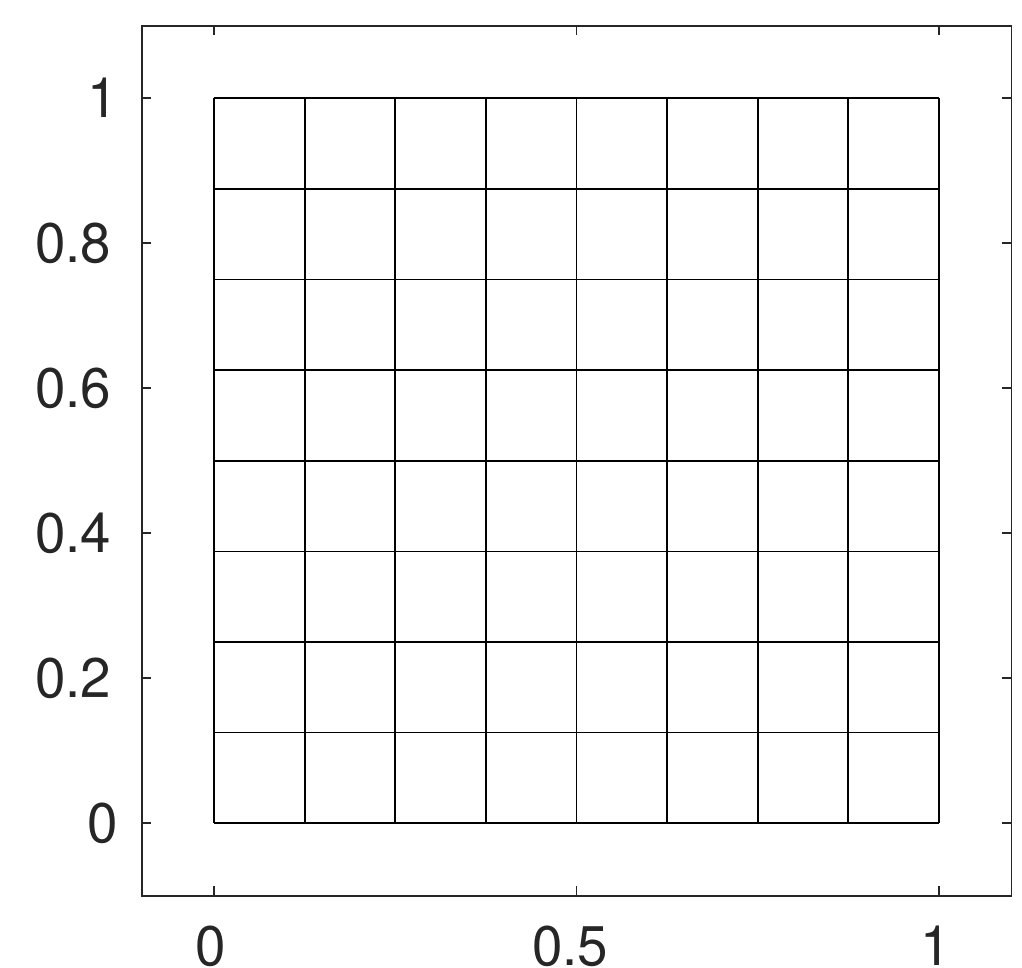}
      \figlab{2d-tc3-grid-quad}
    }
    \subfigure[Triangular elements]{
      \includegraphics[trim=5.2cm 9.0cm 6cm 9.1cm,clip=true,width=0.4\textwidth]{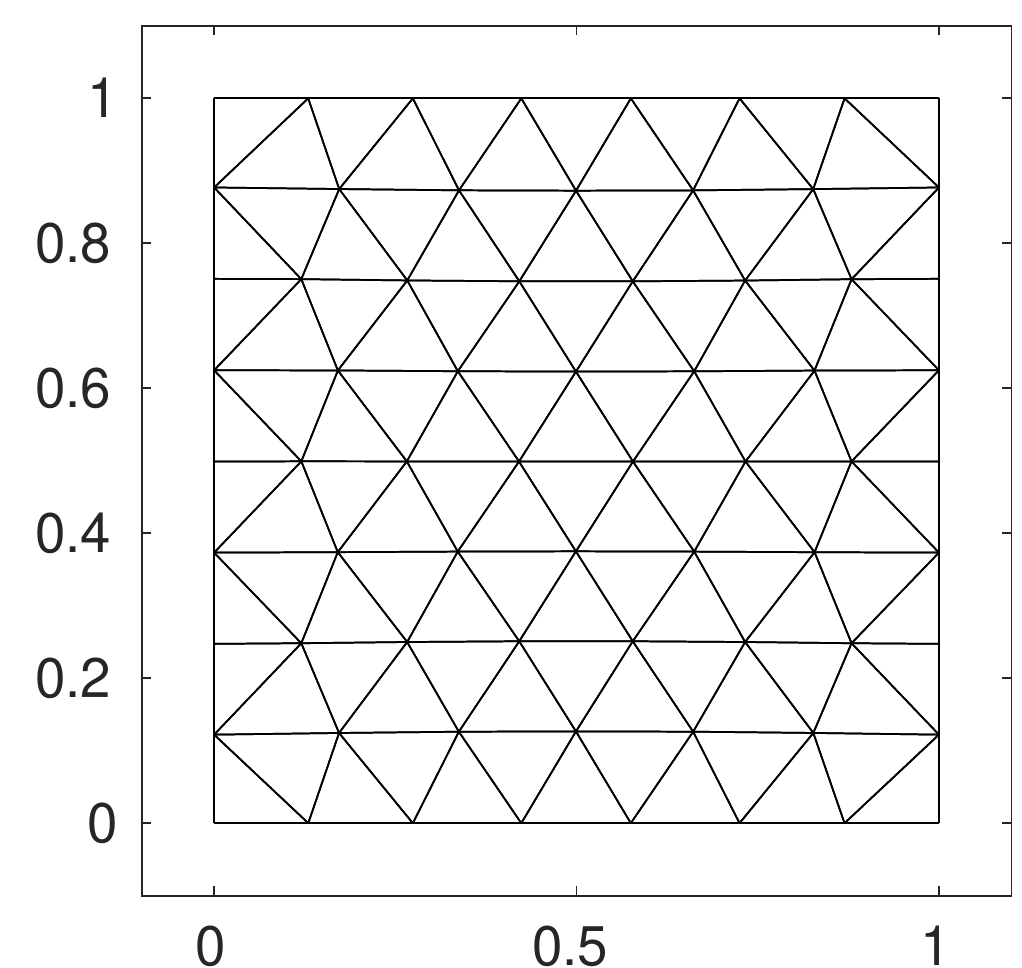}
      \figlab{2d-tc3-grid-triangle}
    }
    \caption{Coarse grids for non-degenerate case with (a) rectangular and (b) triangular elements.}
    \figlab{2d-tc3-grid}
  \end{figure}

Table \tabref{ConvergenceTC3-k23}
 shows  $h$-convergence results in the $L^2(\Omegah)$-norm 
using a sequence of nested meshes with $N_e=\{8^2,32^2,128^2\}$ for rectangular 
and $N_e=\{104,416,1664\}$ for triangular elements, respectively.
The corresponding coarse meshes are shown in Figure \figref{2d-tc3-grid}.
We observe approximately the optimal convergence rates of $\LRp{\k+1}$ for both
scaled pressure $\p$ and scaled velocity $\ub$ for both mesh types.  
\begin{table}[h!t!b!]
\caption{Non-degenerate case: the results show that the HDG solutions
  for scaled pressure $\p$ and scaled velocity $\ub$ converge to the
  exact solution with optimal order of  $\k+1$ for both triangular and rectangular meshes.
  The upwind based parameter $\tau=\phimhalf d$ is used. }
\tablab{ConvergenceTC3-k23}
\begin{center}
\begin{tabular}{*{2}{c}|*{4}{c}|*{5}{c}}
\hline

                     &
& \multicolumn{4}{c|}{Rectangular elements} 
& \multicolumn{4}{c}{Triangular elements} \tabularnewline
\multirow{2}{*}{$\k$} & \multirow{2}{*}{$h$} 
& \multicolumn{2}{c}{$\left\Vert p^e - p     \right\Vert _{\Omegah}$} 
& \multicolumn{2}{c|}{$\left\Vert {\bf u}^e - {\bf u} \right\Vert _{\Omegah}$}
& \multirow{2}{*}{$h$} 
& \multicolumn{2}{c}{$\left\Vert p^e - p     \right\Vert _{\Omegah}$} 
& \multicolumn{2}{c}{$\left\Vert {\bf u}^e - {\bf u} \right\Vert _{\Omegah}$} \tabularnewline
  &  &  error & order & error & order  &  &  error & order & error & order \tabularnewline
\hline\hline
\multirow{4}{*}{1} 
      &   0.0312 &  3.628E-02 &   $-$ & 8.546E-01 & 1.395 &  0.1400 & 6.521E-01 &      $-$ &   5.021E+00 &      $-$\tabularnewline
      &   0.0156 &  1.159E-02 & 1.646 & 2.878E-01 & 1.570 &  0.0700 & 1.813E-01 &    1.847 &   1.436E+00 &    1.806\tabularnewline
      &   0.0078 &  3.389E-03 & 1.775 & 8.804E-02 & 1.709 &  0.0350 & 4.726E-02 &    1.940 &   3.794E-01 &    1.920\tabularnewline
\tabularnewline                                              
\multirow{4}{*}{2}                                           
      &   0.0312 &  1.067E-03 &   $-$ & 2.734E-02 &   $-$ &  0.1400 & 9.445E-02 &      $-$ &   8.312E-01 &      $-$\tabularnewline
      &   0.0156 &  1.597E-04 & 2.741 & 4.272E-03 & 2.678 &  0.0700 & 1.245E-02 &    2.924 &   1.071E-01 &    2.956\tabularnewline
      &   0.0078 &  2.226E-05 & 2.843 & 6.170E-04 & 2.791 &  0.0350 & 1.587E-03 &    2.971 &   1.354E-02 &    2.984\tabularnewline
\tabularnewline                                             
\multirow{4}{*}{3}                                           
      &   0.0312 &  1.970E-05 &   $-$ & 4.691E-04 &  $-$  &  0.1400 & 8.929E-03 &      $-$ &   6.658E-02 &      $-$\tabularnewline
      &   0.0156 &  1.405E-06 & 3.809 & 3.478E-05 & 3.753 &  0.0700 & 5.865E-04 &    3.928 &   4.570E-03 &    3.865\tabularnewline
      &   0.0078 &  9.480E-08 & 3.890 & 2.414E-06 & 3.848 &  0.0350 & 3.728E-05 &    3.976 &   2.942E-04 &    3.957\tabularnewline
\tabularnewline                                            
\multirow{4}{*}{4}                                           
      &   0.0312 &  3.327E-07 &   $-$ & 8.829E-06 &   $-$ &  0.1400 & 8.015E-04 &      $-$ &   7.305E-03 &      $-$\tabularnewline
      &   0.0156 &  1.168E-08 & 4.832 & 3.211E-07 & 4.781 &  0.0700 & 2.571E-05 &    4.963 &   2.292E-04 &    4.995\tabularnewline 
      &   0.0078 &  3.906E-10 & 4.903 & 1.115E-08 & 4.848 &  0.0350 & 8.124E-07 &    4.984 &   7.199E-06 &    4.992\tabularnewline
\hline\hline                                              
\end{tabular}
\end{center}
\end{table}

 

\subsection{Degenerate case with a smooth solution}
\seclab{degenerateSmooth}

Following \cite{arbogast2016linear}
we consider the smooth pressure of the form $\pres^e=\cos(6xy^2)$ on $\Omega=(-1,1)^2$ and the following degenerate porosity 
\begin{align}
  \eqnlab{2d-tc4-porosity}
  \phi = 
  \begin{cases} 
    0, & x \le -\frac{3}{4} \text{ or } y \le -\frac{3}{4},\\
    (x+\frac{3}{4})^\alpha (y+\frac{3}{4})^{2\alpha}, &\text {otherwise}.
  \end{cases}
\end{align}
We note that $\phimhalf \nabla \phi \in \LRs{L^\infty(\Omega)}^2$
for $\alpha\ge 2$, and  we take $\alpha=2$.
The
one-phase region is denoted as $\Omega_1:=\{(x,y): x \le -\frac{3}{4} \text{ or } y \le -\frac{3}{4} \}$ with $\phi=0$, and the two-phase region is given by $\Omega_2:=\{(x,y): -\frac{3}{4} < x < 1 \text{ and } -\frac{3}{4} < y < 1\}$ with $\phi > 0$.
We define the intersection of these two regions by $\pOmega_{12}:=\overline{\Omega}_1 \cap \overline{\Omega}_2$.
In $\overline{\Omega}_1$, the exact  scaled pressure and scaled velocity vanish. 
In $\overline{\Omega}_2$, the exact solutions are given by
  \begin{subequations}
    \eqnlab{2d-tc4-solutions}
      \begin{align}
        \p^e &= \LRp{x + \frac{3}{4}}^{\frac{\alpha}{2}} \LRp{y + \frac{3}{4}}^{\alpha} \cos(6 x y^2),\\
        u_x^e &= 6 y^2 \LRp{x + \frac{3}{4}}^\alpha \LRp{y + \frac{3}{4}}^{2\alpha} \sin \LRp{6 x y^{2}}, \\
        u_y^e &= 12 x y \LRp{x + \frac{3}{4}}^{\alpha} \LRp{y + \frac{3}{4}}^{2\alpha} \sin (6 x y^{2}).
      \end{align}
  \end{subequations}

In Figure
\figref{2d-tc4-ss-pressure} are the contours of the pressure $\pres$ and 
the scaled pressure $\p$ computed from our HDG method using  $N_e=64^2$ rectangular elements and solution order 
$\k=4$. 
We observe that the pressure $\pres$ changes smoothly in the two-phase region $\Omega_2$,
but abruptly becomes zero in the one-phase region $\Omega_1$.
The sudden pressure jump on the intersection $\Omega_{12}$ 
is alleviated with the use of the scaled pressure $\p$. 

\begin{figure}[h!t!b!]
    \centering
    \subfigure[Pressure $\pres$]{
      \includegraphics[trim=4.2cm 1.2cm 24cm 1.0cm,clip=true,width=0.4\textwidth]{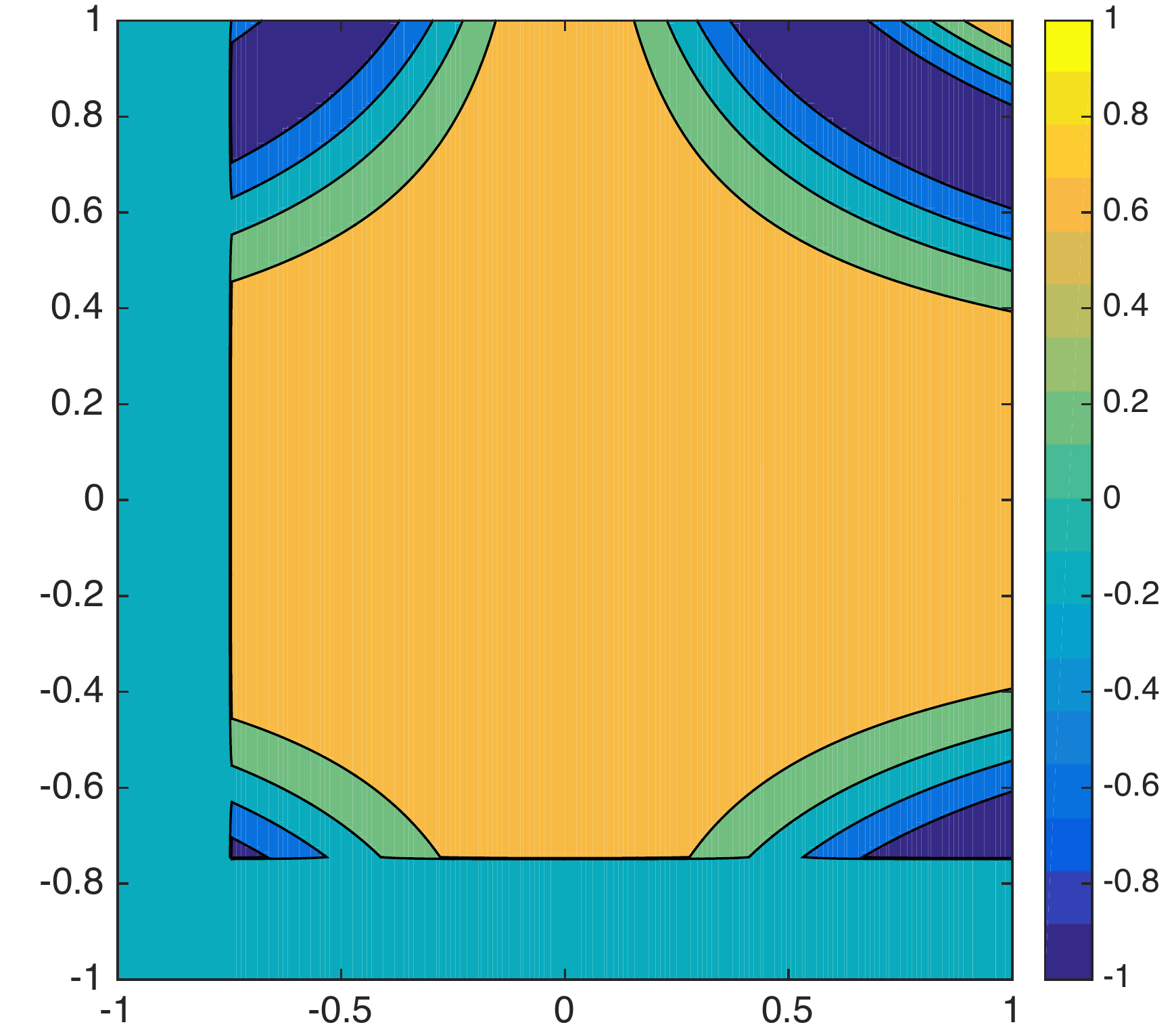}
      \figlab{2d-tc4-ss-pressure-p}
    }
    \subfigure[Scaled pressure $\p$]{
      \includegraphics[trim=24cm 1.2cm 4.2cm 1.15cm,clip=true,width=0.4\textwidth]{./figures/2d-tc4-p4h64-pressure-new}
      \figlab{2d-tc4-ss-pressure-q}
    }
    \caption{Degenerate case with a smooth solution:  (a) contour plot of the pressure $\pres$ field and (b) contour plot of the scaled pressure $\p$ with $N_e=64^2$ and $\k=4$. The pressure field changes smoothly in the two-phase region $\Omega_2$, but suddenly becomes zero in the one-phase region $\Omega_1$. The abrupt change near the intersection $\Omega_{12}$ between the one- and two-phase regions is alleviated with the use of the scaled pressure $\p$.} 
    \figlab{2d-tc4-ss-pressure}
  \end{figure}

For a convergence study, we use a sequence of meshes with $n_e=\{16,32,64,128\}$ and with
$\k=\{1,2,3,4\}$.
Here we choose an even
number of elements so that the mesh skeleton aligns
with the intersection $\pOmega_{12}$.
As can be seen in Figure \figref{2d-tc4-convergence},
the convergence rate of $\LRp{\k+\frac{1}{2}}$ is observed more or less 
 for both the scaled pressure $\p$ and the scaled velocity $\ub$, and this 
agrees with Theorem \theoref{convergence} for the degenerate case with a piecewise smooth solution. 


\begin{figure}[h!t!b!]
    \centering
    \subfigure[$\p$ ]{
      \includegraphics[trim=.0cm .1cm 2.2cm .1cm,clip=true,width=0.4\textwidth]{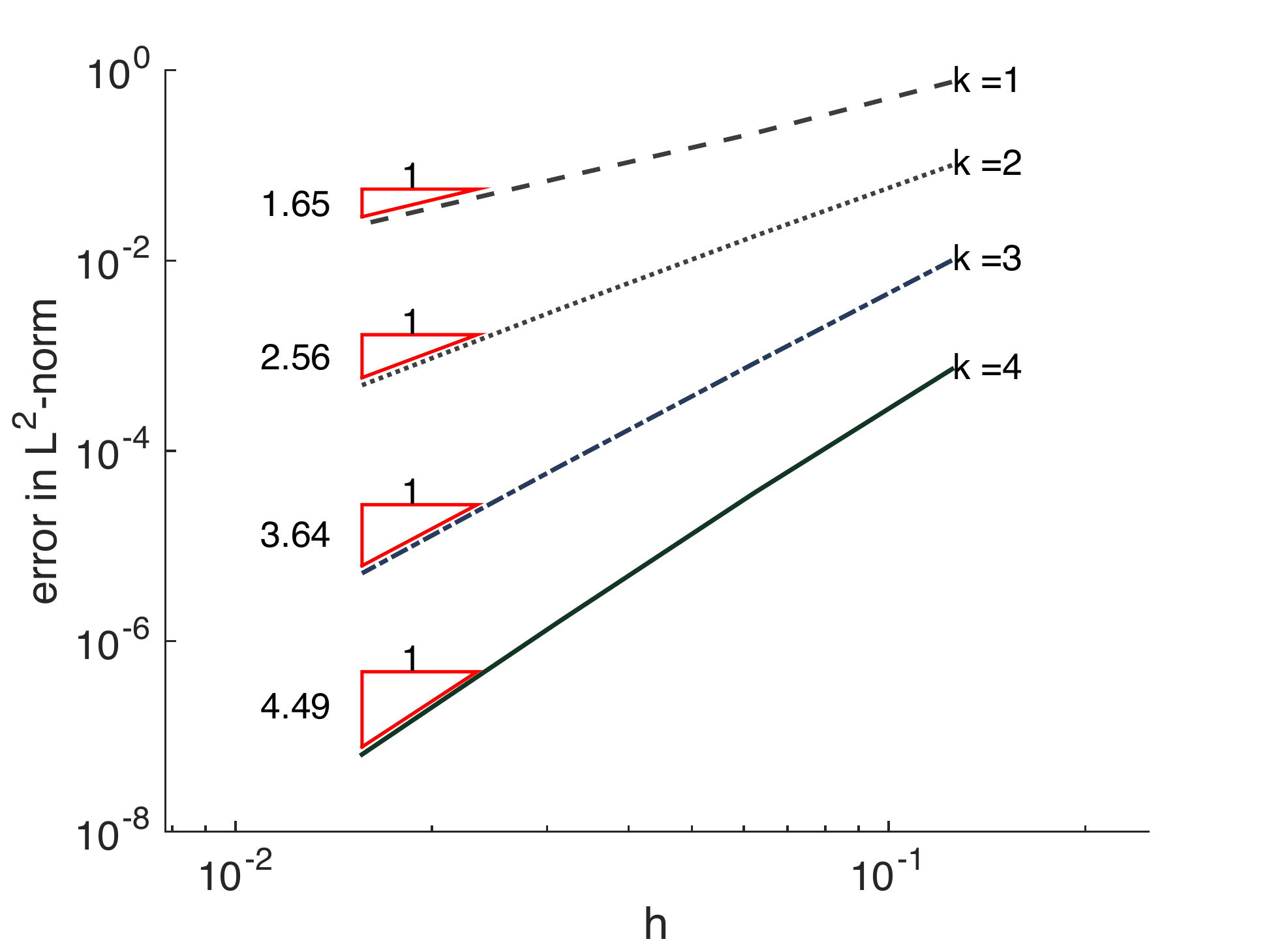}
      \figlab{2d-tc4-conv-tau-q}
    }
    \subfigure[$\ub$]{
      \includegraphics[trim=.1cm .1cm 2.2cm .1cm,clip=true,width=0.4\textwidth]{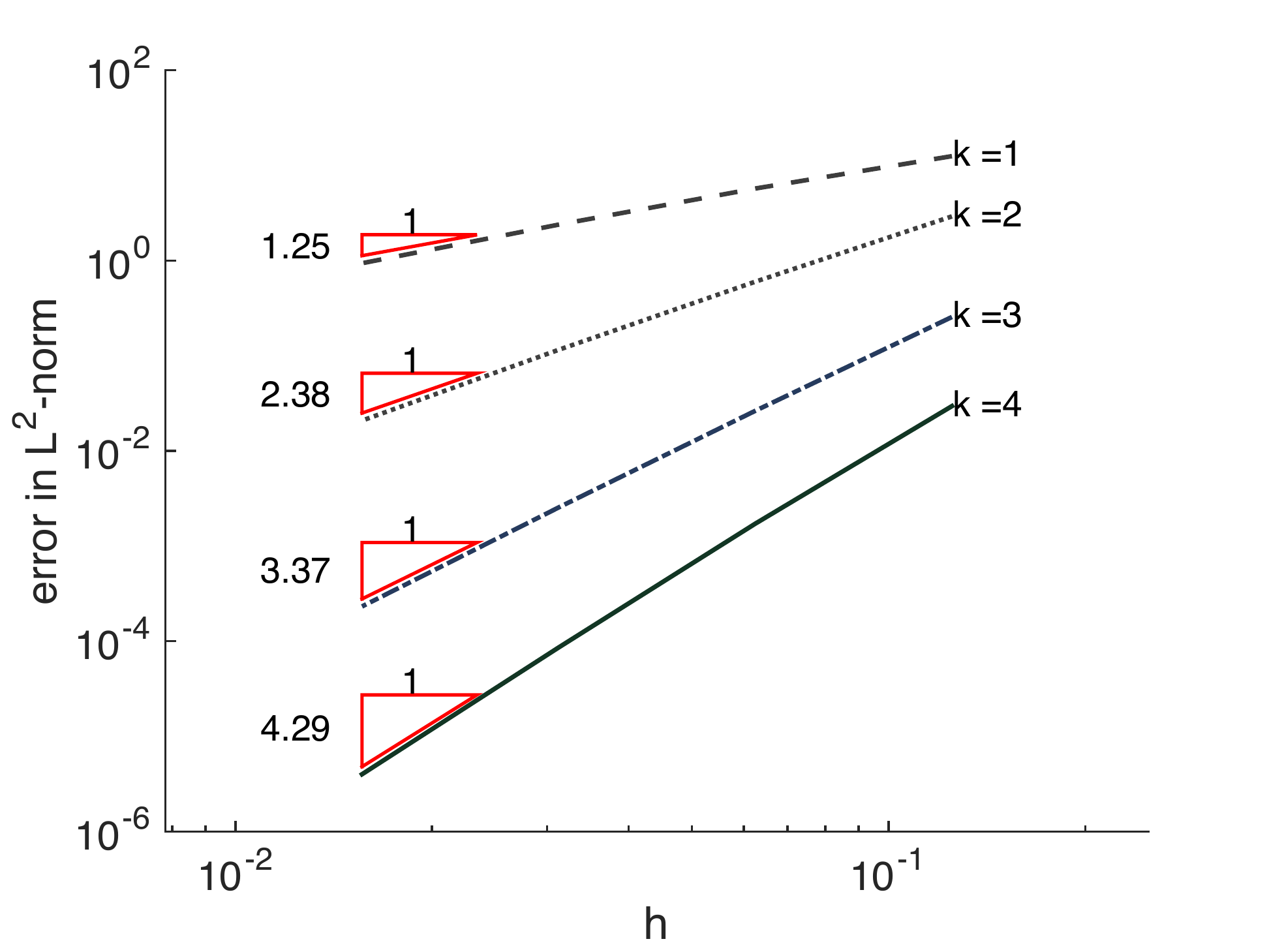}
      \figlab{2d-tc4-conv-tau-u}
    }
    \caption{Degenerate case with a smooth solution: convergence study for (a) the scaled pressure $\p$ field
      and (b) the scaled velocity $\ub$ field.
      The $\LRp{\k+\frac{1}{2}}$ convergence rates 
      are obtained approximately for both the scaled pressure $\p$ and the scaled velocity $\ub$.}
    \figlab{2d-tc4-convergence}
  \end{figure}

\subsection{Degenerate case with low solution regularity}
\seclab{degenerateNonSmooth}
Similar to \cite{arbogast2016linear}, 
we choose the exact pressure to be $\pres^e=y(y-3x)(x+\frac{3}{4})^\beta$ with $\beta=-\frac{1}{4} \text{ or } -\frac{3}{4}$, 
and the porosity $\phi$ is defined in \eqnref{2d-tc4-porosity}.
Similar to Section \secref{degenerateSmooth}, we take $\alpha=2$.
The exact solutions then read
  \begin{subequations}
    \eqnlab{2d-tc5-solutions}
    \begin{align}
      \p^e &= y(y-3x)\LRp{x+\frac{3}{4}}^{\frac{\alpha}{2}+\beta}\LRp{y+\frac{3}{4}}^\alpha,\\
      u_x^e &= y \LRp{ \beta(3x-y) + 3\LRp{x+\frac{3}{4}} } \LRp{x+\frac{3}{4}}^{\alpha +\beta -1 } \LRp{y+\frac{3}{4}}^{2\alpha}, \\
      u_y^e &= (3x-2y)                         \LRp{x+\frac{3}{4}}^{\alpha + \beta   } \LRp{y+\frac{3}{4}}^{2\alpha}.
    \end{align}
  \end{subequations}
The pressure and the scaled pressure fields are simulated with $N_e=64^2$ and $\k=4$ for the two
different cases: $\beta=-\frac{1}{4}$ and $\beta=-\frac{3}{4}$ in
Figure \figref{2d-tc5-ss-pressure}.  
As can be seen from \eqnref{2d-tc5-solutions} and Figure \figref{2d-tc5-ss-pressure} that smaller
$\beta$ implies lower solution regularity. The pressure field with
$\beta=-\frac{3}{4}$ is less regular than that with
$\beta=-\frac{1}{4}$. For both cases, we also observe that the
pressure $\pres$ fields become stiffer (stiff ``boundary layer'') near the
intersection at $x=-\frac{3}{4}$, while the scaled pressure $\p$
fields are much less stiff.

\begin{figure}[h!t!b!]
    \centering

    \subfigure[$\pres$ with $\beta=-\frac{1}{4}$]{
      \includegraphics[trim=4.5cm 1.2cm 24.0cm 1.125cm,clip=true,width=0.22\textwidth]{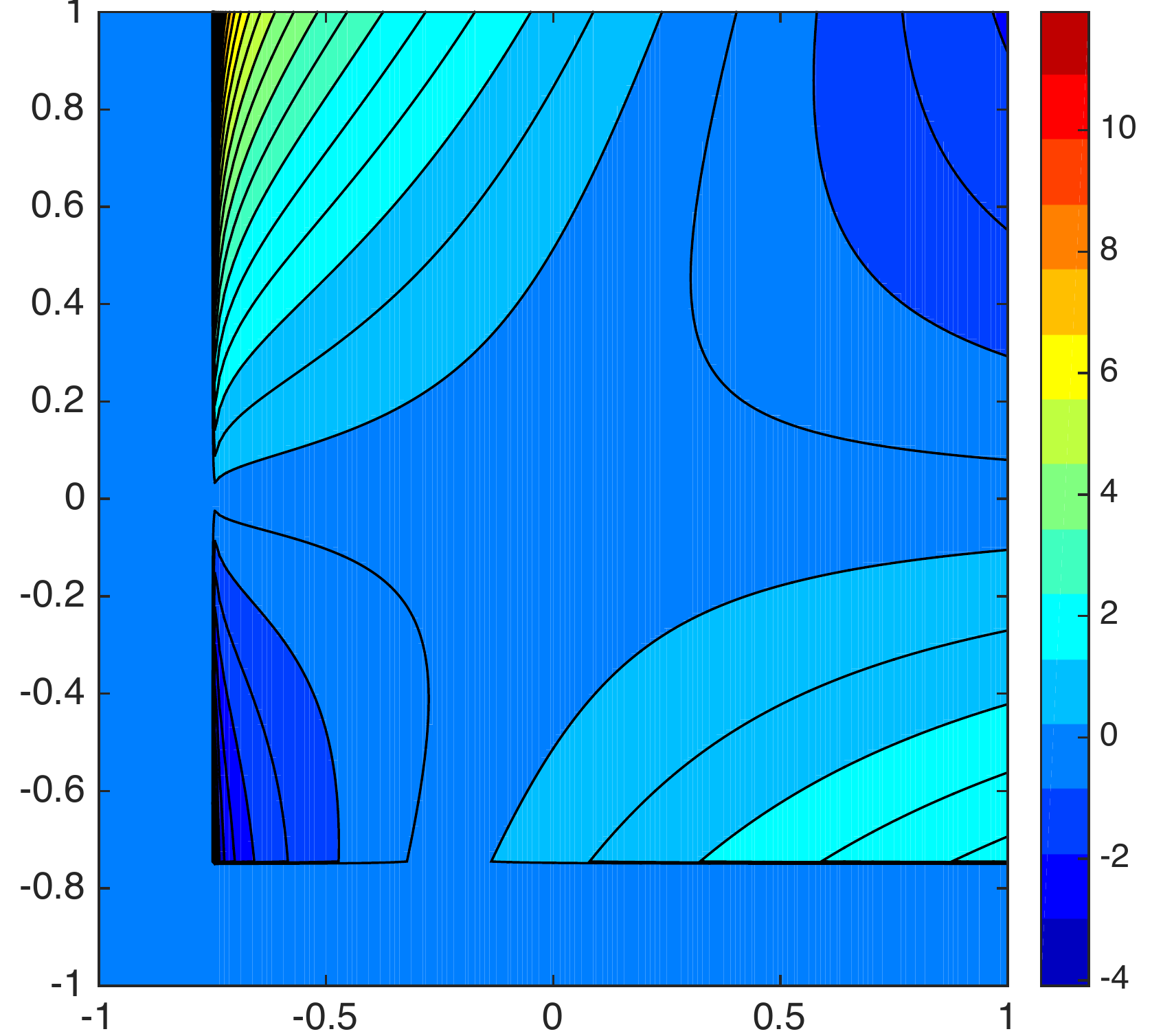}
      \figlab{2d-tc5a-p4h64-pressure-p}
    }
    \subfigure[$\p$ with $\beta=-\frac{1}{4}$]{
      \includegraphics[trim=24cm 1.2cm 4.5cm 1.125cm,clip=true,width=0.22\textwidth]{./figures/2d-tc5a-p4h64-pressure}
      \figlab{2d-tc5a-p4h64-pressure-q}
    }
    \subfigure[$\pres$ with $\beta=-\frac{3}{4}$]{
      \includegraphics[trim=4.5cm 1.2cm 24.0cm 1.125cm,clip=true,width=0.22\textwidth]{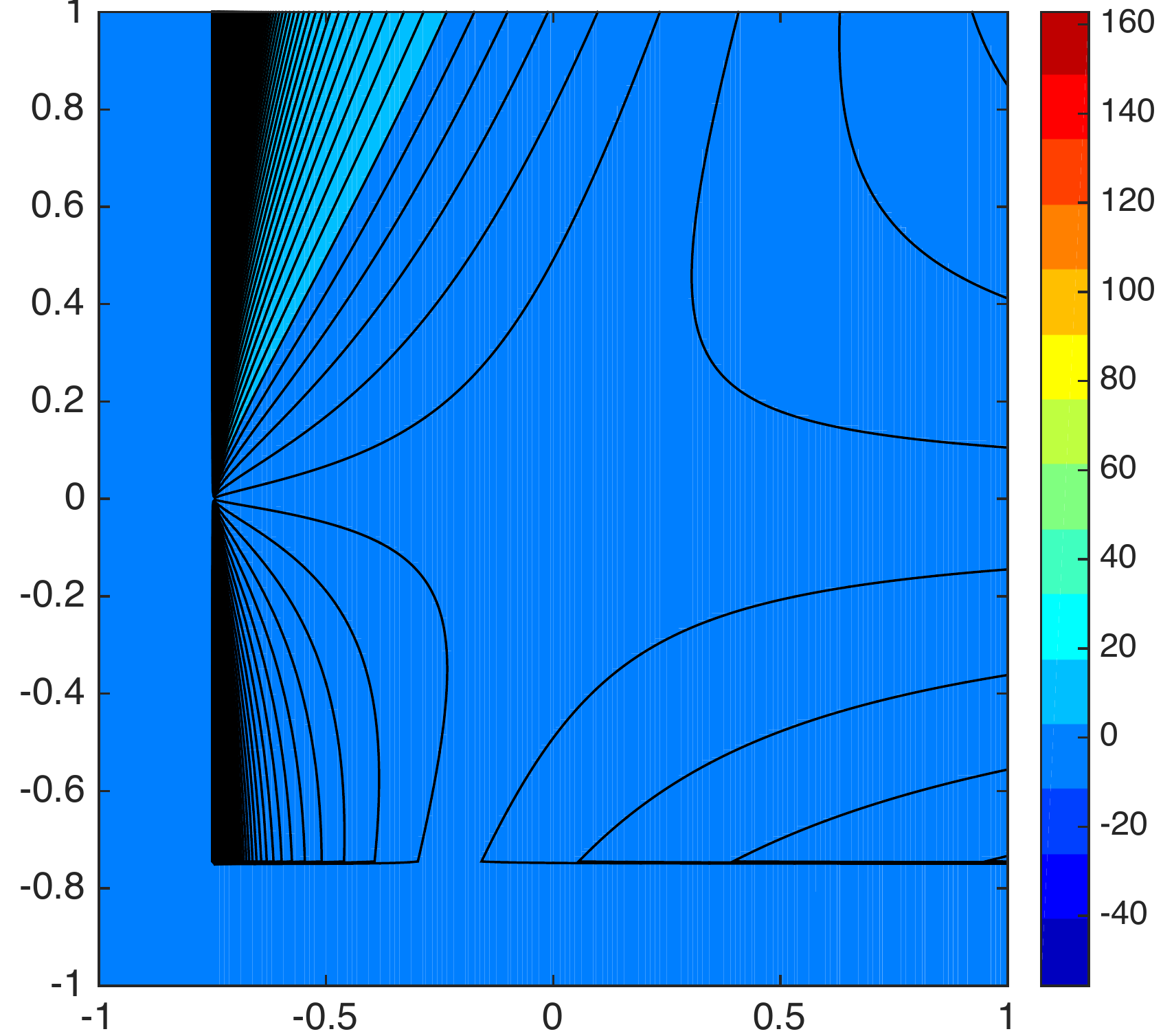}
      \figlab{2d-tc5b-p4h64-pressure-p}
    }
    \subfigure[$\p$ with $\beta=-\frac{3}{4}$]{
      \includegraphics[trim=24cm 1.2cm 4.5cm 1.125cm,clip=true,width=0.22\textwidth]{./figures/2d-tc5b-p4h64-pressure}
      \figlab{2d-tc5b-p4h64-pressure-q}
    }
    \caption{Degenerate case with low solution regularity: simulated with $N_e=64^2$ and $\k=4$ are 
      (a) pressure $\pres$ for $\beta=-\frac{1}{4}$, (b) scaled pressure
      $\p$ for $\beta=-\frac{1}{4}$, (c) pressure $\pres$ with
      $\beta=-\frac{3}{4}$, and (d) scaled pressure $\p$ for
      $\beta=-\frac{3}{4}$.  The pressure field with
      $\beta=-\frac{3}{4}$ is less regular than that with
      $\beta=-\frac{1}{4}$.  In both the cases, the pressure fields
      have low regularity near the intersection $\Omega_{12}$.} 
    \figlab{2d-tc5-ss-pressure}
  \end{figure}

When $\beta=-\frac{1}{4}$, the scaled pressure $\p$ and the scaled velocity $\ub$
reside in $H^{1.25 - \varepsilon}$ for $\varepsilon > 0$ \cite{arbogast2016linear}.
In order to see how the HDG solution behaves for this case, we perform a convergence study with $n_e=\{16,32,64,128\}$ and
$\k=\{1,2,4,8\}$.  As shown in Figure \figref{2d-tc5a-convergence},
the scaled pressure $\p$ and the scaled velocity $\ub$
  converge to the exact counterparts with the rate of about $1.25$. Note that though our error analysis in Section
\secref{errorAnalysis} considers exact solutions residing in standard
Sobolev spaces with integer powers, it can be straightforwardly extended to solutions in
fractional Sobolev spaces. For this example, the convergence rate is
bounded above by $1.25-\varepsilon$ regardless of the solution order.
However, the high order HDG solutions are still beneficial in terms of accuracy, for example,
the HDG solution with $\k=8$ is $4.5$ times more accurate than that with $\k=4$.

\begin{figure}[h!t!b!]
    \centering
    \subfigure[Convergence of $\p$ ]{
      \includegraphics[trim=.1cm .1cm 1.0cm .1cm,clip=true,width=0.4\textwidth]{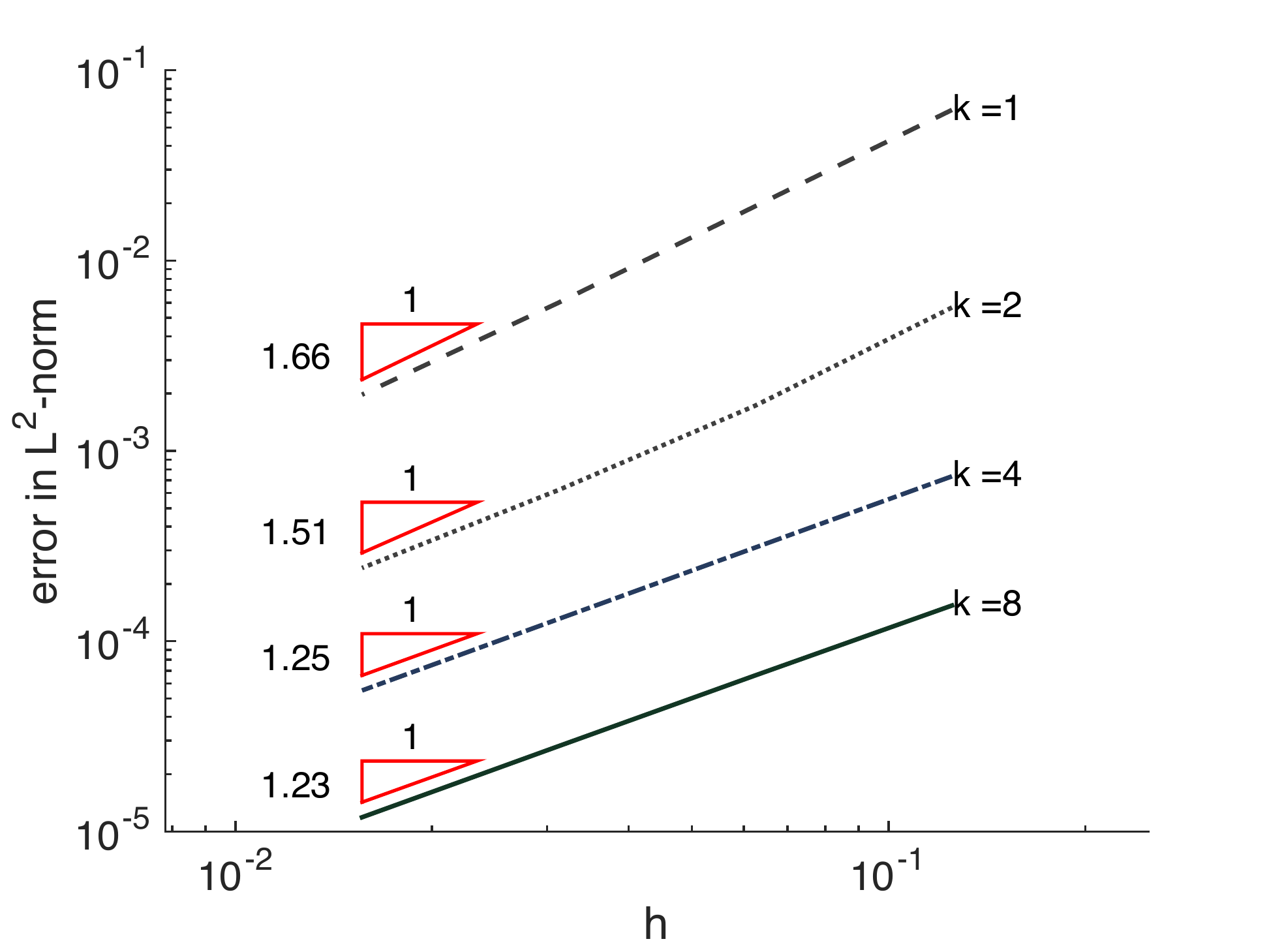}
      \figlab{2d-tc5a-conv-q}
    }
    \subfigure[Convergence of $\ub$ ]{
      \includegraphics[trim=.1cm .1cm 1.0cm .1cm,clip=true,width=0.4\textwidth]{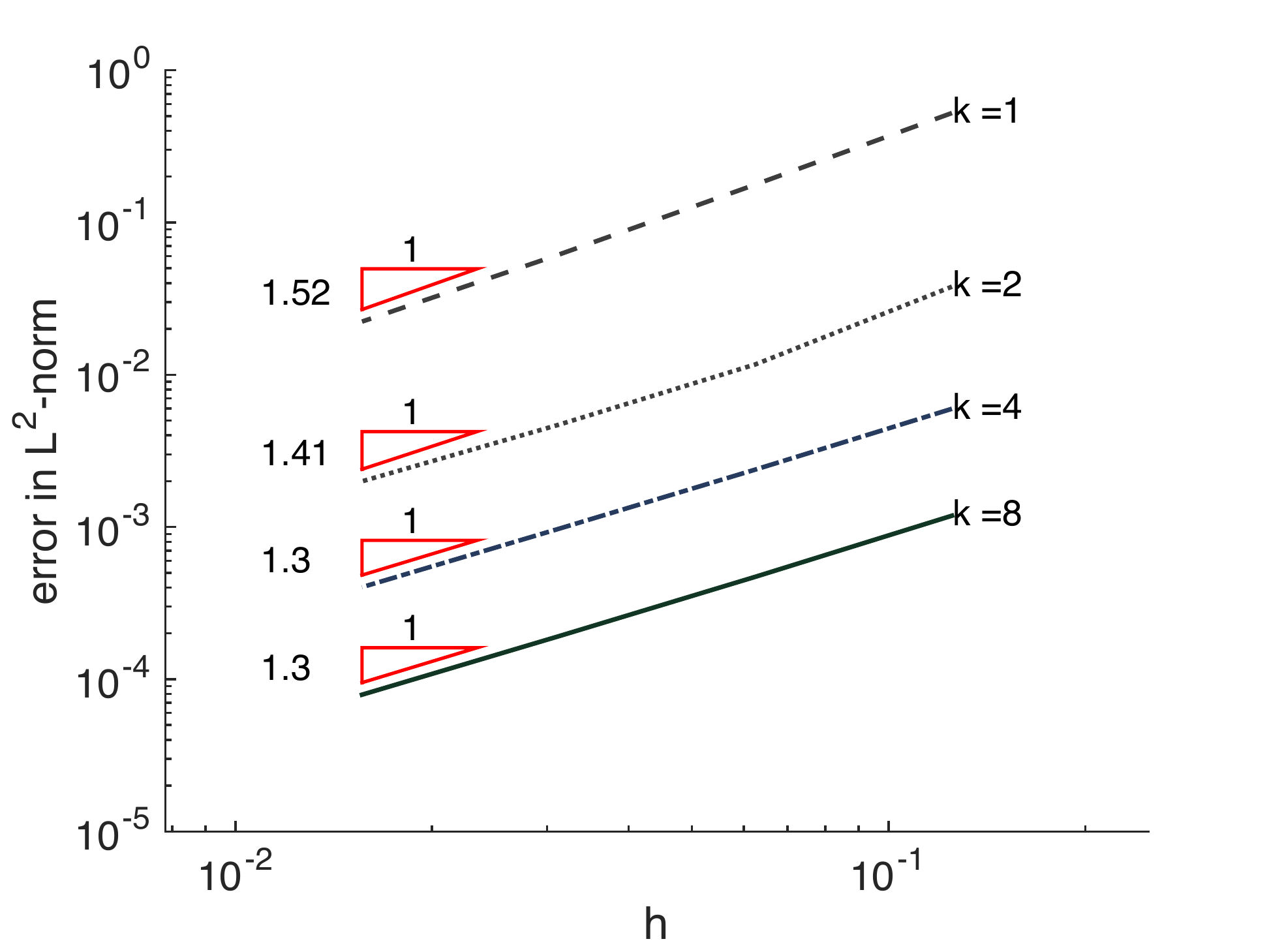}
      \figlab{2d-tc5a-conv-u}
    }
    \caption{Degenerate case with low solution regularity: 
      a convergence study with $\beta=-\frac{1}{4}$ for (a) the scaled pressure $\p$ field and
      (b) the scaled velocity $\ub$ field.}
    \figlab{2d-tc5a-convergence}
\end{figure}


When $\beta=-\frac{3}{4}$, 
the scaled pressure $\p$ and the scaled velocity $\ub$
lie in $H^{0.75-\varepsilon}$ for $\varepsilon>0$ \cite{arbogast2016linear}.
We conduct a convergence study with $n_e=\{16,32,64,128\}$ and
$\k=\{1,2,4,8\}$.  As shown in Figure \figref{2d-tc5b-convergence},
the convergence rate of about $0.75$ is observed for both 
the scaled pressure $\p$ and the scaled velocity $\ub$.
Similar to the case of $\beta=-\frac{1}{4}$,
high order HDG solutions, in spite of more computational demand, are beneficial from an accuracy standpoint.
For instance,
the HDG solution with $\k=8$ is $2.5$ times more accurate than that with $\k=4$.

\begin{figure}[h!t!b!]
    \centering
    \subfigure[Convergence of $\p$ ]{
      \includegraphics[trim=.1cm .1cm 1.0cm .1cm,clip=true,width=0.4\textwidth]{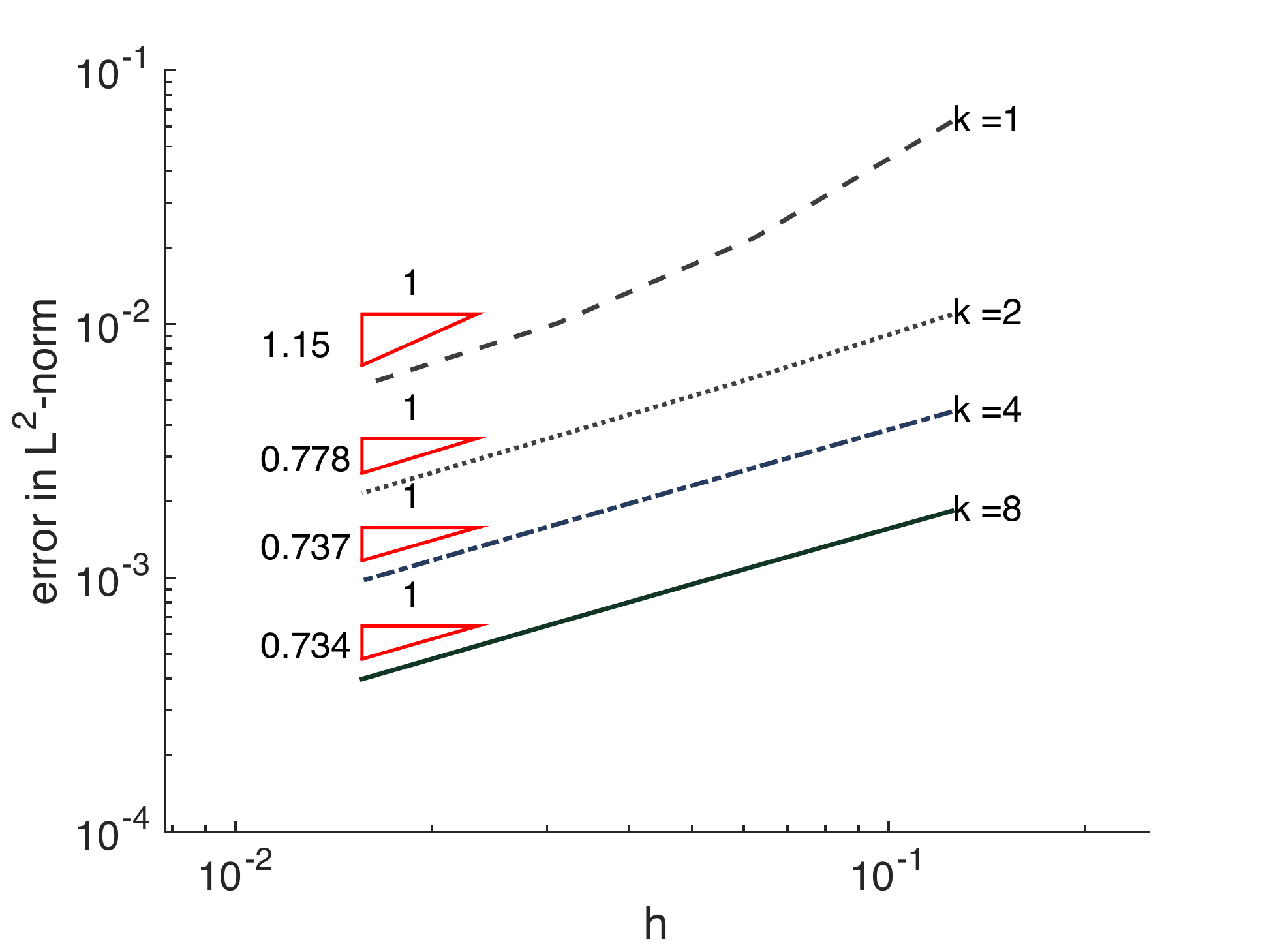}
      \figlab{2d-tc5b-conv-q}
    }
    \subfigure[Convergence of $\ub$]{
      \includegraphics[trim=.1cm .1cm 1.0cm .1cm,clip=true,width=0.4\textwidth]{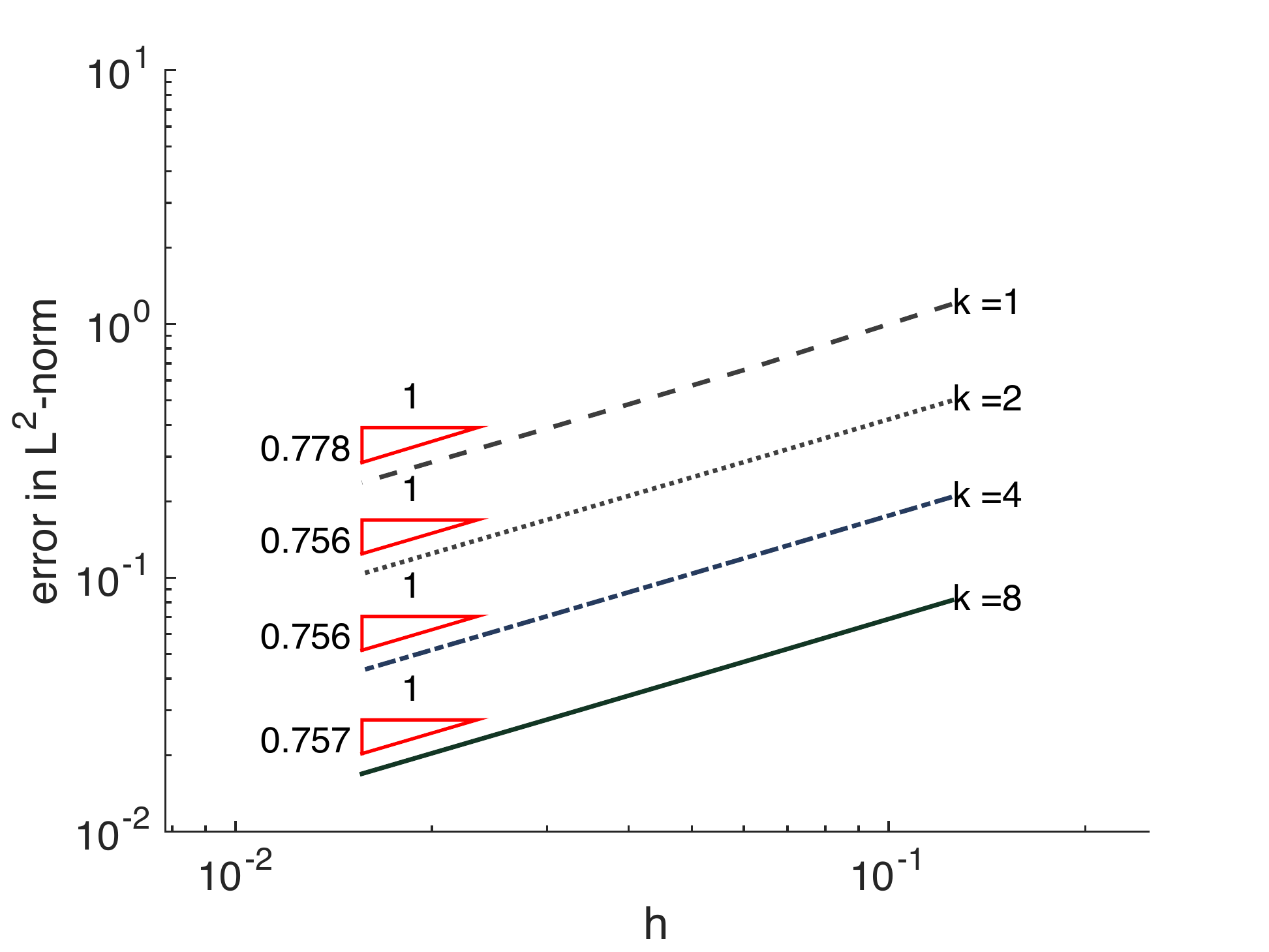}
      \figlab{2d-tc5b-conv-u}
    }
    \caption{Degenerate case with low solution regularity: 
      a convergence study with $\beta=-\frac{3}{4}$ for (a) the scaled pressure $\p$ field and
      (b) the scaled velocity $\ub$ field.}
    \figlab{2d-tc5b-convergence}
  \end{figure}

\subsection{Sensitivity of $\tau$ for the degenerate case with smooth solution}
\seclab{sensitivityTau} In this section we assess numerically whether 
the sub-optimality in Theorem \theoref{convergence} is sharp. To that
end, we consider the degenerate example with smooth solution in
Section \secref{degenerateSmooth} again here. Recall that the
generalized parameter $\tau$ is given by
\begin{equation}
\eqnlab{mixedtau}
\tau:=
\left\{
\begin{array}{ll}
\phimhalf\d & \text{for }  \phi > 0,\\
\gamma      & \text{for } \phi=0.
\end{array}
\right.
\end{equation}
We first compare the convergence rates for three different values of
$\gamma$, namely $\gamma \in \LRc{1/\h, 1,10}$, and the numerical results
(not shown here) show that the convergence rates are the same and are sub-optimal by half order. For that
reason we show only the case when $\gamma = 1/\h$ in the fourth column
of Tables \tabref{SensitivityTest-tau-tc4-pressure} and
\tabref{SensitivityTest-tau-tc4-velocity}, in which we report the convergence rates of $\p$ and $\ub$, respectively. 

We now present convergence rates for the cases where we use a single
value for $\tau$ over the entire mesh skeleton $\Gh$. We consider three cases: $\tau = \LRp{1/\h,1,10}$.
The convergence rates of $\p$ and $\ub$ for these parameters are shown
in the sixth, eighth, and tenth columns of Tables
\tabref{SensitivityTest-tau-tc4-pressure} and
\tabref{SensitivityTest-tau-tc4-velocity}.  The results for $\tau=1$
and $\tau=10$ show the convergence rate of about $\LRp{\k+\half}$.  The
cases with $\tau=\frac{1}{h}$ initially have the convergence rate of
$\LRp{\k+1}$ for both the scaled pressure $\p$ and the scaled velocity
$\ub$, then approach the predicted asymptotic rate of $\LRp{\k+\half}$
as the grid is refined. If we look at the value of the errors at any
grid level, the cases with $\tau = 1/\h$ have the smallest errors
compared to the other cases (including the cases with $\tau$ given in
\eqnref{mixedtau}). It could be due to the initial higher-order convergence
and/or smaller error constants. We thus recommend that $\tau = 1/\h$ should be used.


\begin{table}[h!t!b!]
\caption{Degenerate case with a smooth solution: 
 the errors $\norm{\p^e - \p}_\Omegah$ and the convergence rates for the scaled pressure. Four cases are presented: 
 $\tau$ given \eqnref{mixedtau}, $\tau=\frac{1}{h}$, $\tau=1$ and $\tau=10$.
 }
\tablab{SensitivityTest-tau-tc4-pressure}
\begin{center}
\begin{tabular}{*{2}{c}|*{2}{c}|*{2}{c}|*{2}{c}|*{2}{c}}
\hline
\multirow{2}{*}{$\k$} & \multirow{2}{*}{$h$} 
& \multicolumn{2}{c|}{$ \tau  = \left\{
\begin{array}{ll}
\phimhalf\d & \text{for }  \phi > 0\\
1/\h      & \text{for } \phi=0
\end{array}
\right.$} 
& \multicolumn{2}{c|}{$ \tau=\frac{1}{h} $} 
& \multicolumn{2}{c|}{$ \tau=1 $} 
& \multicolumn{2}{c}{$ \tau=10 $} \tabularnewline
  &  &   error & order  & error & order & error & order & error & order \tabularnewline
\hline\hline
\multirow{4}{*}{ 1} &0.1250 & 7.534E-01 &      $-$ & 1.752E-01 &      $-$ &       2.606E+00 &      $-$ &       3.827E-01 &      $-$ \tabularnewline
                    &0.0625 & 2.188E-01 &    1.784 & 4.107E-02 &    2.093 &       4.809E-01 &    2.438 &       1.228E-01 &    1.639 \tabularnewline
                    &0.0312 & 7.323E-02 &    1.579 & 9.138E-03 &    2.168 &       1.424E-01 &    1.756 &       4.098E-02 &    1.584 \tabularnewline
                    &0.0156 & 2.403E-02 &    1.608 & 2.306E-03 &    1.987 &       5.002E-02 &    1.509 &       1.312E-02 &    1.643 \tabularnewline
\tabularnewline                                      
\multirow{4}{*}{ 2} &0.1250 & 1.004E-01 &      $-$ & 3.313E-02 &      $-$ &       1.642E-01 &      $-$ &       6.442E-02 &      $-$ \tabularnewline
                    &0.0625 & 1.819E-02 &    2.465 & 3.315E-03 &    3.321 &       3.283E-02 &    2.322 &       1.115E-02 &    2.531 \tabularnewline
                    &0.0312 & 3.083E-03 &    2.561 & 3.887E-04 &    3.092 &       6.316E-03 &    2.378 &       1.791E-03 &    2.638 \tabularnewline
                    &0.0156 & 4.907E-04 &    2.651 & 6.161E-05 &    2.658 &       1.154E-03 &    2.453 &       2.713E-04 &    2.723 \tabularnewline
\tabularnewline                                      
\multirow{4}{*}{ 3} &0.1250 & 1.016E-02 &      $-$ & 2.815E-03 &      $-$ &       1.823E-02 &      $-$ &       5.781E-03 &      $-$ \tabularnewline
                    &0.0625 & 8.531E-04 &    3.574 & 1.651E-04 &    4.092 &       1.705E-03 &    3.419 &       4.683E-04 &    3.626 \tabularnewline
                    &0.0312 & 6.857E-05 &    3.637 & 1.098E-05 &    3.910 &       1.540E-04 &    3.469 &       3.630E-05 &    3.690 \tabularnewline
                    &0.0156 & 5.239E-06 &    3.710 & 8.459E-07 &    3.698 &       1.317E-05 &    3.547 &       2.694E-06 &    3.752 \tabularnewline
\tabularnewline                                      
\multirow{4}{*}{ 4} &0.1250 & 7.243E-04 &      $-$ & 2.445E-04 &      $-$ &       1.177E-03 &      $-$ &       4.613E-04 &      $-$ \tabularnewline
                    &0.0625 & 3.700E-05 &    4.291 & 7.196E-06 &    5.086 &       6.741E-05 &    4.126 &       2.298E-05 &    4.327 \tabularnewline
                    &0.0312 & 1.615E-06 &    4.518 & 2.170E-07 &    5.051 &       3.288E-06 &    4.358 &       9.600E-07 &    4.581 \tabularnewline
                    &0.0156 & 6.562E-08 &    4.621 & 8.533E-09 &    4.669 &       1.508E-07 &    4.447 &       3.730E-08 &    4.686 \tabularnewline
\hline\hline
\end{tabular}
\end{center}
\end{table}

\begin{table}[h!t!b!]
\caption{Degenerate case with a smooth solution: 
 the errors $\norm{\ub^e - \ub}_\Omegah$ and the convergence rates for the scaled velocity. Four cases are presented: 
 $\tau$ given \eqnref{mixedtau}, $\tau=\frac{1}{h}$, $\tau=1$ and $\tau=10$.
 }
\tablab{SensitivityTest-tau-tc4-velocity}
\begin{center}
\begin{tabular}{*{2}{c}|*{2}{c}|*{2}{c}|*{2}{c}|*{2}{c}}
\hline
\multirow{2}{*}{$\k$} & \multirow{2}{*}{$h$} 
& \multicolumn{2}{c|}{$ \tau  = \left\{
\begin{array}{ll}
\phimhalf\d & \text{for }  \phi > 0\\
1/\h      & \text{for } \phi=0
\end{array}
\right.$} 
& \multicolumn{2}{c|}{$ \tau=\frac{1}{h} $} 
& \multicolumn{2}{c|}{$ \tau=1 $} 
& \multicolumn{2}{c}{$ \tau=10 $} \tabularnewline
  &  &   error & order & error & order & error & order & error & order \tabularnewline
\hline\hline
\multirow{4}{*}{ 1}&0.1250 &   1.251E+01 &      $-$ & 7.416E+00 &      $-$ &       1.515E+01 &      $-$ &       1.028E+01 &      $-$ \tabularnewline
                   &0.0625 &   5.714E+00 &    1.130 & 2.229E+00 &    1.734 &       7.362E+00 &    1.041 &       4.479E+00 &    1.199 \tabularnewline
                   &0.0312 &   2.386E+00 &    1.260 & 5.664E-01 &    1.977 &       3.334E+00 &    1.143 &       1.784E+00 &    1.329 \tabularnewline
                   &0.0156 &   9.371E-01 &    1.348 & 1.505E-01 &    1.912 &       1.449E+00 &    1.202 &       6.660E-01 &    1.421 \tabularnewline
\tabularnewline                                      
\multirow{4}{*}{ 2}&0.1250 &   2.911E+00 &      $-$ & 1.656E+00 &      $-$ &       3.649E+00 &      $-$ &       2.361E+00 &      $-$ \tabularnewline
                   &0.0625 &   5.996E-01 &    2.279 & 2.372E-01 &    2.804 &       8.361E-01 &    2.126 &       4.566E-01 &    2.371 \tabularnewline
                   &0.0312 &   1.154E-01 &    2.378 & 3.575E-02 &    2.730 &       1.821E-01 &    2.199 &       8.247E-02 &    2.469 \tabularnewline
                   &0.0156 &   2.080E-02 &    2.472 & 5.949E-03 &    2.587 &       3.770E-02 &    2.272 &       1.408E-02 &    2.550 \tabularnewline
\tabularnewline                                      
\multirow{4}{*}{ 3}&0.1250 &   2.551E-01 &      $-$ & 1.237E-01 &      $-$ &       3.595E-01 &      $-$ &       1.876E-01 &      $-$ \tabularnewline
                   &0.0625 &   2.635E-02 &    3.275 & 9.486E-03 &    3.705 &       4.001E-02 &    3.168 &       1.885E-02 &    3.315 \tabularnewline
                   &0.0312 &   2.542E-03 &    3.374 & 6.954E-04 &    3.770 &       4.232E-03 &    3.241 &       1.763E-03 &    3.418 \tabularnewline
                   &0.0156 &   2.316E-04 &    3.456 & 5.436E-05 &    3.677 &       4.236E-04 &    3.321 &       1.565E-04 &    3.494 \tabularnewline
\tabularnewline                                      
\multirow{4}{*}{ 4}&0.1250 &   2.951E-02 &      $-$ & 1.710E-02 &      $-$ &       3.615E-02 &      $-$ &       2.437E-02 &      $-$ \tabularnewline
                   &0.0625 &   1.717E-03 &    4.103 & 6.531E-04 &    4.710 &       2.348E-03 &    3.944 &       1.344E-03 &    4.181 \tabularnewline
                   &0.0312 &   8.585E-05 &    4.322 & 2.274E-05 &    4.844 &       1.321E-04 &    4.152 &       6.358E-05 &    4.402 \tabularnewline
                   &0.0156 &   3.994E-06 &    4.426 & 8.741E-07 &    4.701 &       6.999E-06 &    4.238 &       2.813E-06 &    4.498 \tabularnewline
\hline\hline
\end{tabular}
\end{center}
\end{table}

\subsection{Enhance accuracy by post-processing}
\seclab{postprocessing} 
In this section, we explore the superconvergence property 
for the degenerate elliptic equations.
It is well-known that the HDG methods have a superconvergence property 
\cite{cockburn2012conditions}, i.e., 
the post-processed solution $\pstar$ converges faster than $\p$.
For the standard elliptic system, a post-processed solution $\pstar$ has $(k+2)$ convergence rate. For degenerate equations, the convergence rate, as shown in the numerical results, is dictated by the regularity of the degenerate solutions. For that reason, it is not meaningful to post-process the solution over the entire domain $\Omega$. 
However, over sub-domains with positive porosity, the solution has higher regularity and thus 
the post-processed counterpart is expected to exhibit super-convergence.  
To this end, we seek a new approximation $\pstar \in \Poly^{k+1}(\K)$
by minimizing $ \norm{  \ub 
-\frac{1}{2} \phi^{-\frac{3}{2}} \d \Grad \phi \p 
+ \phimhalf \d \Grad \pstar }^2 $ over an element $\K$, 
which leads to the following local equations:
\begin{subequations}
\eqnlab{postprocessing}
  \begin{align}
  \eqnlab{postprocessing-a}
  \LRp{\Grad \pstar,\Grad \omega }_\K 
    & = - \LRp{\phihalf \d^{-1} \ub,\Grad \omega}_\K 
        + \half \LRp{\phi^{-1}\Grad \phi \p,\Grad \omega }_\K,\\ 
  \LRp{\pstar,1}_\K &= \LRp{\p,1}_K,
  \end{align}
\end{subequations}
for all $\omega \in \Poly^{k+1}(\K)$. 
Since the resulting linear system has $k+2$ equations and any one of \eqnref{postprocessing-a} is a linear combination of the others, 
we remove one of the rows of the linear system of \eqnref{postprocessing-a} in order to obtain a unique solution.
Similarly, we utilize the post-processing technique for the unscaled fluid pressure $\pres$. 
We seek a new approximation $\presstar \in \Poly^{k+1}(\K)$
by minimizing $ \norm{  \ub 
-\d \Grad \presstar }^2 $ over an element $\K$, 
which leads to the following local equations:
\begin{subequations}
\eqnlab{postprocessing_pres}
  \begin{align}
  \LRp{\Grad \presstar,\Grad \omega }_\K 
    & = - \LRp{\d^{-1} \ub,\Grad \omega}_\K,\\
  \LRp{\presstar,1}_\K &= \LRp{\pres,1}_K,
  \end{align}
\end{subequations}
for all $\omega \in \Poly^{k+1}(\K)$. 

Table \tabref{PostProcessing-tc3} shows
$h$-convergence results for the non-degenerate case in Section \secref{simplexQuad} with $(m_x,m_y)=(1,2)$ 
using a series of nested meshes, $N_e = \{8^2,16^2,32^2\}$, for rectangular elements.
We observe that both the post-processed scaled pressure $\pstar$ 
and the post-processed fluid pressure $\presstar$ converge faster than the scaled pressure $\p$ and the fluid pressure $\pres$.
The convergence rates of both the $\pstar$ and $\presstar$ are approximately $(\k+2)$,
except when $\k=1$. 

Next, we examine if the post-processing technique can be utilized in a degenerate case.
To address the question, 
we first define a subdomain $\tilde{\Omega} \subset \Omega$, which is a two-phase region and is ``far'' enough from the degenerate regions 
so that the solutions are less affected by the degeneracy. 
We perform several convergence studies 
for the degenerate cases in Section \secref{degenerateSmooth} and Section \secref{degenerateNonSmooth} 
using a sequence of nested meshes, $N_e = \{16^2,32^2,64^2\}$, for rectangular elements.

Table \tabref{PostProcessing-tc4} shows 
$h$-convergence results for the degenerate case with the smooth solution in Section \secref{degenerateSmooth} 
over $\tilde{\Omega} \in [-0.5,1]^2$. 
In general, both the post-processed solutions  $\pstar$ and $\presstar$ converge faster than the solutions $\p$ and $\pres$.
When $\k>1$, 
the convergence rates for the post-processed solutions $\pstar$ and $\presstar$ are $(\k+\frac{3}{2})$, one order faster than $\p$ and $\pres$. 

For the degenerate case with low solution regularity in Section \secref{degenerateNonSmooth},
the HDG solution already attains the maximal convergence order. It is thus not meaningful to post-process the solution over the whole domain as there will be no gain in the convergence rate.
Instead, we conduct post-processing studies over $\tilde{\Omega} \in [-0.5,1]^2$ and $\tilde{\Omega} \in [0,1]^2$ where the solution is expected to be more regular, and hence allowing the post-processing procedure to achieve a better convergence rate.
In Table \tabref{PostProcessing-tc5b} with $\beta=-\frac{3}{4}$,
 the convergence rate for the scaled pressure $\p$ is approximately $(\k+1)$ when $\k\le 2$, 
 but it starts to degrade to $2.3$ as the order $\k$ increases.
This numerically implies that the exact solution  $\pe$ over $\tilde{\Omega}$ resides in $H^{2.5-\varepsilon}\LRp{\tilde{\Omega}}$ (i.e., it has a higher regularity than over the whole domain). When the convergence rate of $\p$ reaches the maximum possible, there is no improvement for the post-processed counterpart.

Another observation is that the post-processed solution over the sub-domain still provides a benefit in terms of accuracy to a certain extent. 
In Table \tabref{PostProcessing-tc5b} and \tabref{PostProcessing-tc5b-region2},
the post-processed scaled pressure $\pstar$ has a smaller error than 
the scaled pressure $\p$. However, the difference between $\pstar$ and $\p$ becomes negligible as the solution order increases.  
A similar behavior is also seen for the post-processed fluid pressure $\presstar$ and the fluid pressure $\pres$.

\begin{table}[t]
\caption{Non-degenerate case: the errors and the convergence rates 
	for the scaled pressure $\p$, the post-processed scaled pressure $\pstar$, 
	the fluid pressure $\pres$, and the post-processed fluid pressure $\presstar$.
	The post-processed solutions show asymptotically $(\k+2)$ convergence rates for $\k>1$.}
\tablab{PostProcessing-tc3}
\begin{center}
\begin{tabular}{*{2}{c}|*{2}{c}|*{2}{c}|*{2}{c}|*{2}{c}}
\hline
\multirow{2}{*}{$k$} & \multirow{2}{*}{$h$} 
& \multicolumn{2}{c|}{$\left\Vert p^e - p       \right\Vert _{2}$} 
& \multicolumn{2}{c|}{$\left\Vert p^e - p^\star \right\Vert _{2}$} 
& \multicolumn{2}{c|}{$\left\Vert \tilde{p}^e - \tilde{p}       \right\Vert _{2}$} 
& \multicolumn{2}{c}{$\left\Vert \tilde{p}^e - \tilde{p}^\star \right\Vert _{2}$} \tabularnewline
  &  &   error & order & error & order & error & order & error & order \tabularnewline
\hline\hline
  \multirow{3}{*}{1} &   0.1250 &       1.508E-01 &      $-$ &       5.852E-02 &      $-$ &       6.476E-02 &      $-$ &       2.822E-02 &      $-$ \tabularnewline
                     &   0.0625 &       5.014E-02 &    1.589 &       1.245E-02 &    2.233 &       2.005E-02 &    1.692 &       5.928E-03 &    2.251 \tabularnewline
                     &   0.0312 &       1.497E-02 &    1.743 &       2.647E-03 &    2.234 &       5.706E-03 &    1.813 &       1.253E-03 &    2.242 \tabularnewline
\tabularnewline
  \multirow{3}{*}{2} &   0.1250 &       1.337E-02 &      $-$ &       5.001E-04 &      $-$ &       4.612E-03 &      $-$ &       1.970E-04 &      $-$ \tabularnewline
                     &   0.0625 &       2.053E-03 &    2.703 &       3.386E-05 &    3.884 &       6.805E-04 &    2.761 &       1.174E-05 &    4.068 \tabularnewline
                     &   0.0312 &       2.912E-04 &    2.818 &       2.275E-06 &    3.896 &       9.361E-05 &    2.862 &       7.182E-07 &    4.031 \tabularnewline
\tabularnewline
  \multirow{3}{*}{3} &   0.1250 &       6.595E-04 &      $-$ &       1.263E-05 &      $-$ &       2.608E-04 &      $-$ &       3.761E-06 &      $-$ \tabularnewline
                     &   0.0625 &       4.815E-05 &    3.776 &       4.484E-07 &    4.816 &       1.819E-05 &    3.842 &       1.232E-07 &    4.932 \tabularnewline
                     &   0.0312 &       3.289E-06 &    3.872 &       1.523E-08 &    4.880 &       1.209E-06 &    3.911 &       3.987E-09 &    4.950 \tabularnewline
\tabularnewline
  \multirow{3}{*}{4} &   0.1250 &       3.109E-05 &      $-$ &       4.289E-07 &      $-$ &       1.083E-05 &      $-$ &       1.194E-07 &      $-$ \tabularnewline
                     &   0.0625 &       1.113E-06 &    4.804 &       7.568E-09 &    5.825 &       3.731E-07 &    4.859 &       1.940E-09 &    5.944 \tabularnewline
                     &   0.0312 &       3.762E-08 &    4.886 &       1.276E-10 &    5.890 &       1.231E-08 &    4.922 &       3.110E-11 &    5.963 \tabularnewline
\hline\hline
\end{tabular}
\end{center}
\end{table}

\begin{table}[t]
\caption{Degenerate case with a smooth solution:
	the errors and the convergence rates 
	for the scaled pressure $\p$, the post-processed scaled pressure $\pstar$, 
	the fluid pressure $\pres$, and the post-processed fluid pressure $\presstar$.
	The errors are computed over the subdomain $\tilde{\Omega}\in [-0.5,1]^2$. 
	The post-processed solutions show asymptotically $(\k+\frac{3}{2})$ convergence rates for $\k>1$.}
\tablab{PostProcessing-tc4}
\begin{center}
\begin{tabular}{*{2}{c}|*{2}{c}|*{2}{c}|*{2}{c}|*{2}{c}}
\hline
\multirow{2}{*}{$k$} & \multirow{2}{*}{$h$} 
& \multicolumn{2}{c|}{$\left\Vert p^e - p       \right\Vert _{2}$} 
& \multicolumn{2}{c|}{$\left\Vert p^e - p^\star \right\Vert _{2}$} 
& \multicolumn{2}{c|}{$\left\Vert \tilde{p}^e - \tilde{p}       \right\Vert _{2}$} 
& \multicolumn{2}{c}{$\left\Vert \tilde{p}^e - \tilde{p}^\star \right\Vert _{2}$} \tabularnewline
  &  &   error & order & error & order & error & order & error & order \tabularnewline
\hline\hline
  \multirow{3}{*}{1} &   0.0625 &       7.659E-01 &      $-$ &       2.301E-01 &      $-$ &       2.264E-01 &      $-$ &       9.357E-02 &      $-$ \tabularnewline
                     &   0.0312 &       2.211E-01 &    1.792 &       5.174E-02 &    2.153 &       6.588E-02 &    1.781 &       2.052E-02 &    2.189 \tabularnewline
                     &   0.0156 &       7.361E-02 &    1.587 &       1.135E-02 &    2.188 &       2.085E-02 &    1.660 &       4.512E-03 &    2.185 \tabularnewline
\tabularnewline
  \multirow{3}{*}{2} &   0.0625 &       1.005E-01 &      $-$ &       6.934E-03 &      $-$ &       2.742E-02 &      $-$ &       2.829E-03 &      $-$ \tabularnewline
                     &   0.0312 &       1.818E-02 &    2.467 &       5.790E-04 &    3.582 &       4.617E-03 &    2.570 &       2.049E-04 &    3.788 \tabularnewline
                     &   0.0156 &       3.077E-03 &    2.563 &       4.939E-05 &    3.551 &       7.517E-04 &    2.619 &       1.522E-05 &    3.750 \tabularnewline
\tabularnewline
  \multirow{3}{*}{3} &   0.0625 &       1.007E-02 &      $-$ &       2.697E-04 &      $-$ &       2.689E-03 &      $-$ &       9.138E-05 &      $-$ \tabularnewline
                     &   0.0312 &       8.422E-04 &    3.579 &       1.620E-05 &    4.512 &       2.174E-04 &    3.629 &       3.697E-06 &    4.627 \tabularnewline
                     &   0.0156 &       6.758E-05 &    3.639 &       7.168E-07 &    4.498 &       1.703E-05 &    3.674 &       1.597E-07 &    4.533 \tabularnewline
\tabularnewline
  \multirow{3}{*}{4} &   0.0625 &       7.239E-04 &      $-$ &       2.446E-05 &      $-$ &       1.921E-04 &      $-$ &       5.478E-06 &      $-$ \tabularnewline
                     &   0.0312 &       3.672E-05 &    4.301 &       6.056E-07 &    5.336 &       8.891E-06 &    4.434 &       1.279E-07 &    5.421 \tabularnewline
                     &   0.0156 &       1.592E-06 &    4.528 &       1.400E-08 &    5.435 &       3.704E-07 &    4.585 &       2.928E-09 &    5.449 \tabularnewline
\hline\hline
\end{tabular}
\end{center}
\end{table}

\begin{table}[t]
	\caption{Degenerate case with low solution regularity for $\beta = -\frac{3}{4}$: 
	the errors and the convergence rates 
	for the scaled pressure $\p$, the post-processed scaled pressure $\pstar$, 
	the fluid pressure $\pres$, and the post-processed fluid pressure $\presstar$.
	The errors are computed over the subdomain $\tilde{\Omega}\in [-0.5,1]^2$.}
\tablab{PostProcessing-tc5b}
\begin{center}
\begin{tabular}{*{2}{c}|*{2}{c}|*{2}{c}|*{2}{c}|*{2}{c}}
\hline
\multirow{2}{*}{$k$} & \multirow{2}{*}{$h$} 
& \multicolumn{2}{c|}{$\left\Vert p^e - p       \right\Vert _{2}$} 
& \multicolumn{2}{c|}{$\left\Vert p^e - p^\star \right\Vert _{2}$} 
& \multicolumn{2}{c|}{$\left\Vert \tilde{p}^e - \tilde{p}       \right\Vert _{2}$} 
& \multicolumn{2}{c}{$\left\Vert \tilde{p}^e - \tilde{p}^\star \right\Vert _{2}$} \tabularnewline
  &  &   error & order & error & order & error & order & error & order \tabularnewline
\hline\hline
  \multirow{3}{*}{1} &   0.0625 &       5.982E-02 &      $-$ &       1.588E-02 &      $-$ &       6.842E-02 &      $-$ &       4.649E-02 &      $-$ \tabularnewline
                     &   0.0312 &       1.711E-02 &    1.806 &       3.798E-03 &    2.064 &       1.808E-02 &    1.920 &       1.168E-02 &    1.992 \tabularnewline
                     &   0.0156 &       4.627E-03 &    1.887 &       9.559E-04 &    1.990 &       4.674E-03 &    1.952 &       2.938E-03 &    1.992 \tabularnewline
\tabularnewline
  \multirow{3}{*}{2} &   0.0625 &       2.679E-03 &      $-$ &       1.641E-04 &      $-$ &       5.039E-03 &      $-$ &       1.644E-03 &      $-$ \tabularnewline
                     &   0.0312 &       3.654E-04 &    2.874 &       1.971E-05 &    3.058 &       5.641E-04 &    3.159 &       1.068E-04 &    3.944 \tabularnewline
                     &   0.0156 &       4.821E-05 &    2.922 &       3.479E-06 &    2.502 &       6.806E-05 &    3.051 &       8.619E-06 &    3.632 \tabularnewline
\tabularnewline
  \multirow{3}{*}{3} &   0.0625 &       7.697E-05 &      $-$ &       1.337E-05 &      $-$ &       1.761E-04 &      $-$ &       4.230E-05 &      $-$ \tabularnewline
                     &   0.0312 &       5.928E-06 &    3.699 &       3.020E-06 &    2.146 &       1.207E-05 &    3.867 &       5.301E-06 &    2.996 \tabularnewline
                     &   0.0156 &       6.829E-07 &    3.118 &       5.954E-07 &    2.343 &       1.253E-06 &    3.267 &       1.027E-06 &    2.367 \tabularnewline
\tabularnewline
  \multirow{3}{*}{4} &   0.0625 &       4.568E-06 &      $-$ &       4.249E-06 &      $-$ &       9.124E-06 &      $-$ &       7.593E-06 &      $-$ \tabularnewline
                     &   0.0312 &       8.400E-07 &    2.443 &       8.381E-07 &    2.342 &       1.470E-06 &    2.634 &       1.458E-06 &    2.381 \tabularnewline
                     &   0.0156 &       1.689E-07 &    2.314 &       1.689E-07 &    2.311 &       2.920E-07 &    2.331 &       2.920E-07 &    2.320 \tabularnewline
\hline\hline
\end{tabular}
\end{center}
\end{table}

\begin{table}[t]
	\caption{Degenerate case with low solution regularity for $\beta = -\frac{3}{4}$: 
	the errors and the convergence rates 
	for the scaled pressure $\p$, the post-processed scaled pressure $\pstar$, 
	the fluid pressure $\pres$, and the post-processed fluid pressure $\presstar$.
	The errors are computed over the subdomain $\tilde{\Omega}\in [0,1]^2$}
\tablab{PostProcessing-tc5b-region2}
\begin{center}
\begin{tabular}{*{2}{c}|*{2}{c}|*{2}{c}|*{2}{c}|*{2}{c}}
\hline
\multirow{2}{*}{$k$} & \multirow{2}{*}{$h$} 
& \multicolumn{2}{c|}{$\left\Vert p^e - p       \right\Vert _{2}$} 
& \multicolumn{2}{c|}{$\left\Vert p^e - p^\star \right\Vert _{2}$} 
& \multicolumn{2}{c|}{$\left\Vert \tilde{p}^e - \tilde{p}       \right\Vert _{2}$} 
& \multicolumn{2}{c}{$\left\Vert \tilde{p}^e - \tilde{p}^\star \right\Vert _{2}$} \tabularnewline
  &  &   error & order & error & order & error & order & error & order \tabularnewline
\hline\hline
  \multirow{3}{*}{1} &   0.0625 &       4.630E-02 &      $-$ &       1.061E-02 &      $-$ &       2.127E-02 &      $-$ &       7.002E-03 &      $-$ \tabularnewline
                     &   0.0312 &       1.352E-02 &    1.776 &       2.237E-03 &    2.246 &       6.038E-03 &    1.817 &       1.413E-03 &    2.309 \tabularnewline
                     &   0.0156 &       3.710E-03 &    1.865 &       5.230E-04 &    2.097 &       1.640E-03 &    1.881 &       3.148E-04 &    2.166 \tabularnewline
\tabularnewline
  \multirow{3}{*}{2} &   0.0625 &       2.210E-03 &      $-$ &       9.474E-05 &      $-$ &       1.034E-03 &      $-$ &       4.804E-05 &      $-$ \tabularnewline
                     &   0.0312 &       3.055E-04 &    2.854 &       7.706E-06 &    3.620 &       1.388E-04 &    2.898 &       3.983E-06 &    3.592 \tabularnewline
                     &   0.0156 &       4.061E-05 &    2.912 &       6.422E-07 &    3.585 &       1.817E-05 &    2.933 &       3.595E-07 &    3.470 \tabularnewline
\tabularnewline
  \multirow{3}{*}{3} &   0.0625 &       6.260E-05 &      $-$ &       1.584E-06 &      $-$ &       3.369E-05 &      $-$ &       8.669E-07 &      $-$ \tabularnewline
                     &   0.0312 &       4.235E-06 &    3.886 &       1.202E-07 &    3.721 &       2.176E-06 &    3.952 &       1.028E-07 &    3.076 \tabularnewline
                     &   0.0156 &       2.791E-07 &    3.924 &       2.114E-08 &    2.507 &       1.409E-07 &    3.949 &       1.987E-08 &    2.371 \tabularnewline
\tabularnewline
  \multirow{3}{*}{4} &   0.0625 &       1.172E-06 &      $-$ &       1.523E-07 &      $-$ &       8.112E-07 &      $-$ &       1.426E-07 &      $-$ \tabularnewline
                     &   0.0312 &       4.882E-08 &    4.586 &       2.980E-08 &    2.354 &       3.792E-08 &    4.419 &       2.820E-08 &    2.338 \tabularnewline
                     &   0.0156 &       6.110E-09 &    2.998 &       5.978E-09 &    2.317 &       5.703E-09 &    2.733 &       5.646E-09 &    2.321 \tabularnewline
\hline\hline
\end{tabular}
\end{center}
\end{table}

\subsection{Non-degenerate case in three dimensions}

We consider finally a non-degenerate case on $\Omega=(0,1)^3$ with the positive porosity $\phi = \exp(2(x+y+z))$.
%
Let the pressure $\pres^e = \sin(m_x \pi x) \sin (m_y \pi y)\sin (m_z \pi z)
\exp\LRp{-(x+y+z)}$.
The corresponding manufactured scaled solutions are given as 
  \begin{subequations}
    \eqnlab{3d-tc3-solutions}
      \begin{align}
        q^e &= \sin(m_x \pi x) \sin (m_y \pi y) \sin (m_z \pi z),\\
      u_x^e &= \exp(x+y+z) \sin (m_y \pi y) \sin (m_z \pi z) \LRp{\sin(m_x \pi x) - m_x \pi \cos(m_x\pi x)},\\
      u_y^e &= \exp(x+y+z) \sin (m_x \pi x) \sin (m_z \pi z ) \LRp{\sin(m_y \pi y) - m_y \pi \cos(m_y\pi y)},\\
      u_z^e &= \exp(x+y+z) \sin (m_x \pi x) \sin (m_y \pi y ) \LRp{\sin(m_z \pi z) - m_z \pi \cos(m_z\pi z)}.
      \end{align}
  \end{subequations}

Table \tabref{NonDegenerateCase3D-k111}
 shows  $h$-convergence results in the $L^2(\Omegah)$-norm 
using a sequence of nested meshes with $n_e=\{8,12,16,20\}$.
Here we take $m_x=m_y=m_z=1$ and use the upwind based parameter $\tau= \phimhalf d$. 
We observe the convergence rates between $\LRp{\k+\half}$ and $\LRp{\k+1}$ for both
scaled pressure $\p$ and scaled velocity $\ub$. Recall that the optimal convergence rate of $\k + 1$ is proved for only simplices, though similar results for quadrilaterals and hexahedra are expected. Indeed, Table \tabref{NonDegenerateCase3D-k111} shows that as the solution order increases, the convergence rate is above $\k+\half$. 

\begin{table}[h!t!b!]
\caption{Non-degenerate case in three dimensions: the results show that the HDG solutions for scaled pressure $\p$ and scaled velocity $\ub$ converge to the exact solutions with the rate in above $k + \frac{1}{2}$.
The upwind based parameter $\tau=\phimhalf d$ is used. }
\tablab{NonDegenerateCase3D-k111}
\begin{center}
\begin{tabular}{*{2}{c}|*{2}{c}|*{2}{c}}
\hline
\multirow{2}{*}{$\k$} & \multirow{2}{*}{$h$} 
& \multicolumn{2}{c|}{$\left\Vert p^e - p     \right\Vert _{2}$} 
& \multicolumn{2}{c}{$\left\Vert {\bf u}^e - {\bf u} \right\Vert _{2}$} \tabularnewline
  &  &   error & order & error & order \tabularnewline
\hline\hline
\multirow{4}{*}{ 1}    
     &   0.1250 &       2.553E-02 &      $-$ &       5.895E-01 &      $-$ \tabularnewline
     &   0.0833 &       1.473E-02 &    1.356 &       3.527E-01 &    1.267 \tabularnewline
     &   0.0625 &       9.754E-03 &    1.433 &       2.399E-01 &    1.340 \tabularnewline
     &   0.0500 &       6.994E-03 &    1.491 &       1.757E-01 &    1.396 \tabularnewline
\tabularnewline
\multirow{4}{*}{ 2}    
     &   0.1250 &       1.405E-03 &      $-$ &       4.408E-02 &      $-$ \tabularnewline
     &   0.0833 &       5.274E-04 &    2.417 &       1.712E-02 &    2.333 \tabularnewline
     &   0.0625 &       2.576E-04 &    2.491 &       8.579E-03 &    2.402 \tabularnewline
     &   0.0500 &       1.460E-04 &    2.545 &       4.961E-03 &    2.455 \tabularnewline
\tabularnewline
\multirow{4}{*}{ 3}    
     &   0.1250 &       4.872E-05 &      $-$ &       1.477E-03 &      $-$ \tabularnewline
     &   0.0833 &       1.183E-05 &    3.491 &       3.684E-04 &    3.425 \tabularnewline
     &   0.0625 &       4.234E-06 &    3.572 &       1.348E-04 &    3.495 \tabularnewline
     &   0.0500 &       1.885E-06 &    3.626 &       6.105E-05 &    3.550 \tabularnewline
\tabularnewline
\multirow{4}{*}{ 4}    
     &   0.1250 &       1.100E-06 &      $-$ &       2.741E-05 &      $-$ \tabularnewline
     &   0.0833 &       1.705E-07 &    4.598 &       4.428E-06 &    4.496 \tabularnewline
     &   0.0625 &       4.447E-08 &    4.672 &       1.188E-06 &    4.573 \tabularnewline
     &   0.0500 &       1.551E-08 &    4.720 &       4.235E-07 &    4.622 \tabularnewline
\hline\hline
\end{tabular}
\end{center}
\end{table}

\section{Conclusions and future work}
\seclab{Conclusion}
In this paper, we are interested in developing numerical methods for both glacier dynamics and mantle convection.
Both the phenomena can be described by a two-phase mixture model,
 in which the mixture of the fluid and the solid is described by the porosity $\phi$ 
(e.g., $\phi >0$ implies the fluid-solid two-phase and $\phi=0$ means the solid single-phase regions). 
The challenge is when the porosity vanishes because the system degenerates, which make the problem difficult to solve numerically.  
To address the issue, following \cite{arbogast2016linear}, we start by scaling variables to obtain the well-posedness.
Then we spatially discretize the system using the upwind HDG framework. 
The key feature is that we have modified the upwind HDG flux to accommodate the degenerate (one-phase) region. 
When the porosity vanishes, the resulting HDG system becomes ill-posed
because the stabilization parameter associated with the HDG flux disappears.
For this reason, we introduce the generalized stabilization parameter 
that is composed of the upwind based parameter $\tau=\phimhalf \d$ in the two-phase region and a positive parameter $\tau=\frac{1}{h}>0$ in the one-phase region.
This enables us to develop a high-order HDG method for a linear degenerate elliptic equation 
arising from a two-phase mixture of both glacier dynamics and mantle convection. 

We have shown the well-posedness and the convergence analysis of our HDG scheme.
The rigorous theoretical results tell us that our HDG method has the convergence rates of $\LRp{\k+1}$ for a non-degenerate case 
and $\LRp{\k+\half}$ for a degenerate case with piecewise smooth solution. 

  

 Several numerical results confirm that our proposed HDG method works well for the linear degenerate elliptic equations.
For the non-degenerate case, we obtain the $\LRp{\k+1}$ convergence rates of both the scaled pressure $\p$ and the scaled velocity $\ub$ in two dimensions, whereas in three dimenions we observe the convergence rates above $\LRp{\k+\half}$.
For the degenerate case with a smooth solution, the convergence rate of $\LRp{\k+\half}$ is observed for both the scaled pressure $\p$ and the scaled velocity $\ub$.
For the degenerate case with low solution regularity, the convergence rates of the numerical solutions are bounded by the solution regularity, 
but the high-order method still shows a benefit in terms of accuracy. 
Through the parameter study, 
we found that the positive parameter on one-phase region does not affect the accuracy of a numerical solution.
We also found that the generalized stabilization parameter show slightly better performance in terms of error levels than the other single-valued parameters for the degenerate case with low solution regularity.

In order for our proposed method to work in two-phase flows, 
the interfaces between matrix solid and fluid melt need to be identified 
and grids should be aligned with the interfaces.
In other words, the degeneracies are always required to lie on a set of measure zero. 
Note that we do not consider the case of partially melted regions yet.
We will tackle this challenge in our future work.

\appendix

\section{Auxiliary results}
\seclab{auxiliary}
In this appendix we collect some technical results that are useful for our analysis.
\begin{lemma}[Inverse Inequality {\cite[Lemma 1.44]{PietroErn12}}] \lemlab{inverse-inequality}
  For $v \in \poly{k}(\K)$ with $\K \in \Omegah$, there exists $c>0$ independent of $h$ such that
  \begin{align}
    \eqnlab{lemmaA1}
    \nor{\Grad v}_{0,\K} \le c h_K^{-1} \nor{v}_{0,\K} .
  \end{align}
\end{lemma}
\begin{lemma}[Trace inequality {\cite[Lemma 1.49]{PietroErn12}}] \lemlab{cont-inv-trace}
  For $v \in H^1(\Omegah)$ and for $\K \in \Omegah$ with $e \subset \pK$, there exists $c>0$ independent of $h$ such that
  \begin{align}
    \nor{v}_{0,e}^2 \le c \LRp{\nor{\Grad v}_{0,\K} + h_{\K}^{-1} \nor{v}_{0,\K} } \nor{v}_{0,\K} .
  \end{align} 
\end{lemma}

Applying the arithmetic-geometric mean inequality to the right side, we can derive 
\begin{align} 
  \nor{v}_{0,e} \le c \LRp{h_k^\half \nor{\Grad v}_{0,\K} + h_{\K}^{-\half} \nor{v}_{0,\K} }.
\end{align}

If $v \in H^1(\Omegah)$ is in a piecewise polynomial space, we can derive the following 
inequality from Lemma~\lemref{cont-inv-trace} and the inverse inequality (Lemma~\lemref{inverse-inequality}):
\begin{equation}
  \eqnlab{discrete_trace_ineq}
  \nor{v}_{0,\e} \le c  { h_{\K}^{-\half} \nor{v}_{0,\K} }.
\end{equation}

\bibliographystyle{unsrt}
\bibliography{ceo}

\end{document}